\documentclass[a4paper,8pt]{article}

\usepackage{a4wide}
\usepackage{authblk}
\usepackage{cite}
\usepackage[british]{babel}
\usepackage[utf8]{inputenc}
\usepackage{xcolor}
\usepackage{graphicx,subfigure}
\usepackage{epstopdf}
    \graphicspath{{./}{figures/}}
\usepackage[colorlinks=true,linkcolor=blue,citecolor=blue]{hyperref}
\usepackage{xcolor}
\usepackage{amsfonts,amsmath,amsthm}

\newtheorem{theorem}{Theorem}[section]
\newtheorem{proposition}[theorem]{Proposition}

\theoremstyle{remark}

\newcommand{\be}{\begin{equation}}
\newcommand{\ee}{\end{equation}}

\newcommand{\erf}{\textrm{erf}}
\newcommand{\R}{\mathbb{R}}
\newcommand{\x}{\mathbf{x}}
\newcommand{\e}{\epsilon}
\newcommand{\rev}{\textcolor{black}}

\newcommand{\fer}[1]{(\ref{#1})}
\allowdisplaybreaks

\begin{document}
\title{Fokker-Planck modeling of many-agent systems in swarm manufacturing: asymptotic analysis and numerical results}

\author[1,3]{Ferdinando Auricchio \thanks{\texttt{ferdinando.auricchio@unipv.it}}} 
\author[2,3]{Giuseppe Toscani \thanks{\texttt{giuseppe.toscani@unipv.it}}}
\author[2]{Mattia Zanella \thanks{\texttt{mattia.zanella@unipv.it}}}

\affil[1]{Department of Civil Engineering and Architecture, University of Pavia, Italy} 
\affil[2]{Department of Mathematics ``F. Casorati'', University of Pavia, Italy} 
\affil[3]{IMATI "E. Magenes", National Research Council (CNR), Pavia, Italy}

\date{}

\maketitle

\begin{abstract}
In this paper we study a novel Fokker-Planck-type model that is designed to mimic manufacturing processes through the dynamics characterizing a large set of agents. In particular, we describe a many-agent system interacting with a target domain in such a way that each agent/particle is attracted by the center of mass of the target domain with the aim to uniformly cover this zone. To this end, we first introduce a mean-field model with discontinuous flux whose large time behavior is such that the steady state is globally continuous and uniform over a connected portion of the domain. We prove that a diffusion coefficient guaranteeing that a a given portion of mass enters in the target domain exists and that it is unique. Furthermore, convergence to equilibrium in 1D is provided through a reformulation of the initial problem involving a nonconstant diffusion function. The extension to 2D is explored numerically by means of recently introduced structure preserving methods for Fokker-Planck equations. 
\medskip


\end{abstract}

\tableofcontents
\section{Introduction}

Manufacturing consists in goods transformation and industrial equipment production, generating 16\% of the world's-GDP (gross domestic product) and almost 15\% of the European GDP \footnote{World Bank national accounts data, and OECD National Accounts data files. Accessed 29 Mar 2022 }. However, despite the huge effort in digitalisation (e.g., industry 4.0), today's manufacturing is still a sequential, human-in-the-loop process, optimised just for well-defined conditions and well-established supply chains. The resulting process is neither resilient nor robust for extreme or unforeseeable environmental conditions.

Along this line of thinking Additive Manufacturing (AM) technologies have introduced significant flexibility in production processes, also establishing the concept of nearly unconstrained design freedom; however, AM still have with limitations, in particular  in the presence of unforeseeable conditions or with respect to production of large components, which are still out of reach.

Swarm manufacturing may then represents an even more significant change in perspective and an exciting field that might overcome limitations of the currently available AM technologies by means of highly flexible and robust solutions. The core idea beneath swarm manufacturing is to combine traditional robotic tools with swarm intelligence, extending manufacturing towards more distributed, hence more robust, processes. The potential applications of such a technology are \rev{manifold}, from aerospace to naval industry, from civil engineering to space constructions; as an example, this technology can be used to produce large structural parts, too complex or even impossible to realise by means of standard processes. 
The application in industrial processes of principles of self-organizing systems represent a very fascinating field towards the practical implementation of swarm robotics, where a swarm of manufacturing interacting agents are designed for the production of complex components, see e.g. \cite{A,CH,H,HW}. Among the main \rev{advantage} of this new approach, the flexibility and robustness in production processes constitute a key \rev{features} when dealing with  complex tasks subject to unpredictable changes.  Indeed, the swarms are typical examples of distributed intelligence where possible defections of part of the swarm does not weaken the assigned control protocols. Furthermore, they have recognized potential for the in-site construction in remote or inhospitable areas, in which traditional equipments cannot be carried in \cite{Ac}. Recently, several research teams are trying to apply the paradigms of swarming to 3D printing technologies, see e.g. \cite{Ox}. However, the design and the control of a swarm of moving, autonomous robots is not a trivial task. 
Therefore, models can help the designer in predicting the swarm behaviour given a predefined environment and a set of process constraints.
Accordingly, the development of effective models is of fundamental importance.

In recent years, we witnessed an explosion of new applications of kinetic and mean-field equations describing the collective behavior of large particles' systems. Particular attention has been paid to self-organizing systems and their emerging patterns in  biology, social sciences, traffic dynamics and robotics. Without intending to review the huge literature on these topics, we point the \rev{readers} to \cite{C,CPZ,CS,DOB,GHR,V}. In these works Langevin-type equations are designed to produce several patterns stemming from interaction forces. Typical examples are synchronization, milling, flocking and swarming behavior and their shape can be suitably modified through external controls \cite{CFPT,FPR}.

In particular, Vlasov-Fokker-Planck-type equations have been developed to represent the evolution of the particles' distribution. These equations can be derived from microscopic particle dynamics in the limit of a large number of agents, see e.g. \cite{AP,CFRT,CFTV,DM,HT,MT,PT,PTZ} and the references therein. Indeed, it is worth to mention that microscopic modelling approaches for interacting systems suffer from the so-called curse of dimensionality, becoming rapidly unaffordable when the number of agents becomes large. On the other hand, mean-field and kinetic models represent a very effective strategy for the dynamics of aggregate quantities.  

In the present work, we concentrate on  Fokker-Planck-type models suitable to describe the action of a large swarm on a fixed domain $D \subseteq \mathbb R^d$, with $d \ge 1$. The simple task of the swarm is to spread uniformly over its surface, which can be interpreted as the deposition of a single layer in standard additive manufacturing technologies. This goal can be classically obtained by  resorting  to a Fokker--Planck equation with a constant diffusion term and a suitable drift which senses the distance from the boundaries of the target domain $D$ and, as soon as it enters in $D$, each particle of the swarm starts a random exploration of the domain. A direct inspection of the model allows to easily verify that the explicit asymptotic distribution profile is a weighted combination of a uniform distribution inside $D$ and a Gaussian distribution in $D\setminus\mathbb R^d$. 
However, the unusual presence of a steady--state solution with a constant value inside $D$ makes it difficult to understand whether and at what time rate this stationary solution can actually be reached.

As we shall see in the forthcoming Section \ref{sec:FP},  the \emph{ad-hoc} drift operator in the Fokker--Planck equation generating the desired equilibrium profile is obtained from a potential that is not convex, so that the actual existing mathematical results on the exponential convergence towards equilibrium for the classical collisional equations \cite{T1,TV,OV}, and for the Fokker-Planck equation with irregular coefficients  \cite{LLB}, are no more valid.  To guarantee convergence together with a uniform covering of the domain $D$, we thus adopt a different path that is based on a reformulation of the problem in terms of a new Fokker-Planck equation that, while possessing the same stationary distribution of the previous one, is based on a different balance between the drift and diffusion terms, where the drift is now derived by a quadratic potential, and  the diffusion term has  a variable diffusion coefficient \cite{FPTT22}. In dimension one, a suitable regularization of the new diffusion function permits then to rigorously prove, by means of a Chernoff type inequality \cite{FPTT},  that the solution to the new Fokker--Planck equation converges towards the steady state solution at a rate that is, at least, polynomial in time. We mention recent efforts in this direction in case of subcritical confinement potential \cite{TZ}.

The extension to 2D problems is to date an open problem and we address this question by means of a numerical test based on recently introduced structure preserving numerical methods \cite{PZ} that guarantee positivity of solutions and an accurate approximation of the large time trends. 

In more detail, the paper is organized as follows. In Section \ref{sec:FP} we introduce a classical Fokker-Planck-type model for manufacturing, mimicking the interaction of a system of agents with a given portion of the domain. We then rigorously study the structure of the equilibrium profile  together with its relevant features in one and two dimensions.  Section \ref{sec:trends} will be devoted  to the  study of the convergence of the solution towards equilibrium. To this aim we will resort to a new Fokker--Planck type equation possessing the same steady state but different drift and diffusion operators. Explicit results are established in the 1D case together with trends characterizing convergence. Finally, in Section \ref{sec:num} we propose several numerical examples in 1D and 2D to test the features of the model and the  convergence rate.  \rev{The consistency of the proposed numerical approach is tested through the comparison with classical particle methods, namely by investigating the convergence to equilibrium of a system of particles whose distribution corresponds to the introduced new Fokker-Planck equation.}

\section{Fokker-Planck model in manufacturing processes}
\label{sec:FP}

We are interested in a system of $N\gg 1$ particles interacting with a space domain $D \subset \mathbb R^d$. To simplify computations, and to concentrate on the main mathematical properties of the model, we start our analysis by assuming that $D$ is a $d$-dimensional sphere of center $\x_0$ and radius $\delta>0$.  Hence $D = \{\x \in \mathbb R^d: |\x-\x_0|\le \delta\}$, where $|\x-\mathbf{y}|\ge 0$ is the Euclidean distance between the points $\x,\mathbf{y} \in \mathbb R^d$.  The system of particles  is such that each particle senses the direction of motion, moves towards the center of the sphere $D$, and it starts to randomly explore the target domain $D$ uniformly  as soon as the particle enters in it. 



We indicate by $f(\x,t)\,d\x$ the probability  of finding a particle in the elementary volume $d\x $ around the point $\x\subset\R^d$ at time $t \ge 0$. The mesoscopic model translating the aforementioned dynamics can be described by a Fokker-Planck-type equation with a constant diffusion  and a discontinuous drift  \cite{Risken}, which can be suitably written in divergence form as
\begin{equation}
\label{eq:model}
\partial_t f(\x,t) = \nabla_\x \cdot \left[ \psi(\x,\x_0) f(\x,t) +{\sigma^2} \nabla_\x f(\x,t) \right].
\end{equation}
In equation \fer{eq:model} the function $\psi$ characterizing the drift is expressed by
\be\label{eq:drift}
\psi(\x,\x_0) = 
\begin{cases}
0 & |\x-\x_0| <\delta \\
\x-\x_0 & |\x-\x_0|\ge \delta, 
\end{cases}
\ee
 Consequently, particles move subject to the simultaneous presence of the drift and diffusion operators, unless they are in the target spherical domain $D$, where only the diffusion operator survives. It is immediate to show that the solution to equation \fer{eq:model} is mass preserving. Thus, without lack of generality, we can assume that the initial distribution $f_0(\x,t)$ is a probability density, so that, at any time $t \ge0$
 \be\label{eq:mass}
 \int_{\R^d} f(\x, t) \, d\x = 1.
 \ee
The (unique) steady state of unit mass of the \rev{Fokker--Planck-type} equation \fer{eq:model} is formally obtained by solving the differential equation
\begin{equation}
\label{eq:equi}
 {\sigma^2} \nabla_\x f(\x,t) = -\psi(\x,\x_0) f(\x,t).
\end{equation}
As we shall see in details in the next section, the solution to \fer{eq:equi} is given by 
\begin{equation}
\label{eq:equi2}
f^\infty(\x)  = 
\begin{cases}
\dfrac{m_1}{(2\pi\sigma^2)^{d/2}}\exp\left\{ - \dfrac{(\x-\x_0)^2}{2\sigma^2}\right\} &\textrm{if}\; |\x-\x_0|\ge \delta \vspace{0.25cm} \\ 
m_2 \left(\dfrac{\delta^d \pi^{d/2}}{\Gamma\left(d/2 +1\right)}\right)^{-1} & \textrm{if}\; |\x-\x_0|< \delta
\end{cases}
\end{equation}
The values $m_1,m_2>0$ are not new parameters and later on it will be showed that \rev{they} are determined by imposing that the total mass of the steady state $f^\infty(\x)$ is equal to $1$,  together with the continuity of the steady state at the boundaries of the target domain $D$. 

The steady state is a  continuous function resulting from the weighted combination of a Gaussian density (outside $D$) and a uniform density (inside $D$).  Formally, the same equilibrium configuration \fer{eq:equi2}  can be obtained 
by resorting to other Fokker--Planck type equations \cite{FPTT}. Among others, one is the following
\begin{equation}
\label{eq:model2}
\partial_t f(\x,t) = \nabla_\x \cdot \left[ (\x-\x_0) f(\x,t) + \nabla_\x \left( \kappa(\x)f(\x,t)\right) \right].
\end{equation}
In equation \fer{eq:model2} the continuous function $\kappa$ characterizing the diffusion coefficient is expressed by
\be
\label{eq:k}
\kappa(\x) = 
\begin{cases}
\sigma^2+ \dfrac{\delta^2}{2} - \dfrac{1}{2}|\x-\x_0|^2 & |\x-\x_0|< \delta \\
\sigma^2 & |\x-\x_0|\ge\delta.
\end{cases}
\ee 
\rev{The introduced choice of non-constant diffusion function  \eqref{eq:k}  is such that the steady states of equations   \eqref{eq:model} and  \fer{eq:model2} are equal, since the differential equation  \fer{eq:equi}  can be equivalently rewritten as 
\[
(\x-\x_0)f(\x,t) + \nabla_\x(\kappa(\x)f(\x,t)) = ((\x-\x_0) + \nabla_\x \kappa(\x))f(\x,t) + \kappa(\x)\nabla_\x f(\x).
\] }
Hence, at variance with \fer{eq:model}, the Fokker--Planck equation describes a system of particles   such that each particle senses the direction of motion, moves towards the center of the sphere $D$, and it starts to randomly explore the target domain $D$ adapting its diffusion to the distance from the center $x_0$ of the domain, according to \fer{eq:k}, as soon as the particle enters in it. 

From the mathematical point of view, the main difference between the two models is that in equation \fer{eq:model2} the drift is obtained from a strongly convex potential, while in \fer{eq:model} it is not.  \rev{Hence, as we shall detail in the next Section, the convergence to the equilibrium configuration \fer{eq:equi2} of the solution to this second model rests on sound and rigorous mathematical results \cite{FPTT}.}

\subsection{Large times behavior}\label{sub:1D}

We start our analysis by studying the main mathematical properties of the common equilibrium of the Fokker--Planck type equations \fer{eq:model}--\fer{eq:model2} in the one dimensional case $d=1$, and we will deal with the relaxation process of the solution towards it. Further, we will discuss the possibility to extend the analysis to the 2D case. 

\subsection{1D case}
Let  $d = 1$. Then, the stationary distribution of \eqref{eq:model} solves the  differential equation
\[
\psi(x,x_0)f + \sigma^2 \partial_x f= 0, 
\]
or equivalently
\[
\begin{cases}
(x-x_0)f^\infty(x) + \sigma^2 \partial_xf^\infty(x) = 0, & x \notin D \\
 \sigma^2 \partial_xf^\infty(x) = 0, & x \in D. 
\end{cases}
\]
Hence, if $x \notin D$, there exists a portion of mass of the equilibrium distribution that behaves like a Gaussian density, solution of the classical Fokker-Planck equation with quadratic potential and constant diffusion.
On the other hand,  if $x \in D$ the steady state distribution is a constant. 
In full generality, the continuous equilibrium distribution can be written as the one-dimensional version of \fer{eq:equi2}, expressed by
\begin{equation}
\label{eq:finf}
f^\infty(x)  = 
\begin{cases}
\dfrac{m_1}{\sqrt{2\pi\sigma^2}}\exp\left\{ - \dfrac{(x-x_0)^2}{2\sigma^2}\right\} &\textrm{if}\; |x-x_0|\ge \delta,\vspace{0.25cm} \\ 
\dfrac{m_2}{2\delta} & \textrm{if}\; |x-x_0|< \delta.
\end{cases}
\end{equation}
The values $m_1,m_2>0$ quantify the percentages of mass of the Gaussian and uniform parts, respectively.
It is immediate to verify that these values  can be determined by imposing mass conservation and continuity of the steady state at the boundaries of the target domain.  From  condition  \fer{eq:mass} we get
\[
\dfrac{m_1}{\sqrt{2\pi\sigma^2}} \int_{|x-x_0|\ge \delta}\exp\left\{ - \dfrac{(x-x_0)^2}{2\sigma^2}\right\}  dx +m_2 = 1, 
\]
which can be rewritten in terms of the \emph{erf} function as
\[
 m_1\left(1-  \erf\left(  \dfrac{\delta}{\sqrt{2\sigma^2}}\right)\right) = 1 - m_2.
 \]
Next, by imposing continuity of the equilibrium density at the boundaries of $D$ gives the condition
\[
 \dfrac{m_1}{\sqrt{2\pi\sigma^2}} \textrm{exp}\left\{-\dfrac{\delta^2}{2\sigma^2}\right\} = \dfrac{m_2}{2\delta}.
\]
Therefore, for any fixed diffusion coefficient $\sigma^2$ and any $\delta>0$ we get precise information on the mass fraction entering in $D$ by solving the linear system 
\begin{equation}
\label{eq:sys}
\begin{cases}
 m_1\left(1-  \erf\left(  \dfrac{\delta}{\sqrt{2\sigma^2}}\right)\right) +m_2= 1  \\
 \dfrac{m_1}{\sqrt{2\pi\sigma^2}} \textrm{exp}\left\{-\dfrac{\delta^2}{2\sigma^2}\right\} - \dfrac{m_2}{2\delta} = 0. 
\end{cases}
\end{equation}
with respect to the pair $(m_1,m_2)$. Since
\[
K = 
\textrm{det}
\begin{bmatrix}
1- \erf\dfrac{\delta}{\sqrt{2\sigma^2}} & 1 \\
\dfrac{1}{\sqrt{2\pi\sigma^2}}\exp\left\{ -\dfrac{\delta^2}{2\sigma^2}\right\} & -\dfrac{1}{2\delta}
\end{bmatrix}
<0
\]
for any $\delta>0$, system \fer{eq:sys} has a unique solution. Furthermore, the values of the solution to \eqref{eq:sys} are positive, coherently with their introduction in \eqref{eq:finf}. 

A further interesting question stemming from system \fer{eq:sys} is related to the possibility to fix in advance the percentage $m_2$ of mass of the steady state located in $D$, and to find consequently the value of $\sigma$ which does the job. Indeed, when considering the system of interacting agents, this corresponds to control in advance the goal of the manufacturing process.  The question is essential to control possible errors in exploring the target domain. 

To this aim, let us fix $m_2$ and let us rewrite system \fer{eq:sys} in terms of the unknown $m_1$ and  $x= \frac{\delta}{\sqrt{2\sigma^2}}$.  System \eqref{eq:sys}  is equivalent to 
\be
\label{eq:sys2}
\begin{cases}
m_1\left[ 1-\erf(x) \right] = 1-m_2, \\
\dfrac{2m_1}{\sqrt{\pi}} x\exp\{-x^2\} = m_2. 
\end{cases}
\ee
Hence, letting $m_* = 1-m_2$ from the first equation of \eqref{eq:sys2} we have 
\[
m_1 = \dfrac{m_*}{1-\erf(x)}
\]
and substituting this value on the second equation we get
\[
\dfrac{2}{\sqrt{\pi}}\dfrac{x\exp\{-x^2\}}{1-\erf(x)} = \dfrac{1-m_*}{m_*}. 
\]
We may observe that the function on the left-hand side of the last equation is  strictly increasing. Indeed, if $H(x) = \dfrac{2}{\sqrt{\pi}}\dfrac{x\exp\{-x^2\}}{1-\erf(x)} $ we have 
\[
\lim_{x \rightarrow 0^+} H(x) = 0, \qquad \lim_{x \rightarrow +\infty} H(x) = +\infty
\]
and 
\[
\begin{split}
\dfrac{d}{dx}H(x) &= \underbrace{\dfrac{2}{\sqrt{\pi}}\dfrac{\exp\{-x^2\}}{(1-\erf(x))^2}}_{> 0} \underbrace{\left[ (1-2x^2)(1-\erf(x)) + \dfrac{2}{\sqrt{\pi}}\exp\{-x^2\}\right] }_{G(x)},
\end{split}\]
where $G(x)$ is a positive strictly decreasing function such that $G(0) = 1$ and $\lim_{x \rightarrow +\infty}G(x) = 0$. Indeed we easily observe that 
\[
\dfrac{d}{dx}G(x) = -4x(1-\erf(x))\le 0
\]
for all $x \in \mathbb R_+$. Since $H(x)$ is a strictly monotone increasing function in $[0,+\infty)$, there exists a unique value $\bar x \in (0,+\infty)$ such that $H(\bar x) = \dfrac{1-m_*}{m_*}$. Therefore, $\sigma^2$ and $m_1>0$ are uniquely determined as 
\[
\sigma^2= \dfrac{\delta^2}{2\bar x^2}, \qquad m_1 = \dfrac{1-m_2}{1-\erf(\bar x)}.
\]
In Figure \ref{fig:mass} we depict the steady states defined in \eqref{eq:sys2} for several choices of $m_2$ and where we considered as target domain $D = \left\{x \in \mathbb R: |x-x_0|\le \frac{1}{2}\right\}$ with $x_0 = 0$. On the right plot we report the obtained diffusion coefficient $\sigma^2$, for several choices of $m_2$, solution to \eqref{eq:sys2}. We can observe how $\sigma^2$ rapidly decays for increasing values of $m_2$. This behavior illustrates that, to guarantee prescribed mass values inside $D$ and given error levels, represented by the tails produced by the distribution outside $D$, it is possible to prescribe uniquely  the diffusion  coefficient.

\begin{figure}
\centering
\includegraphics[scale = 0.375]{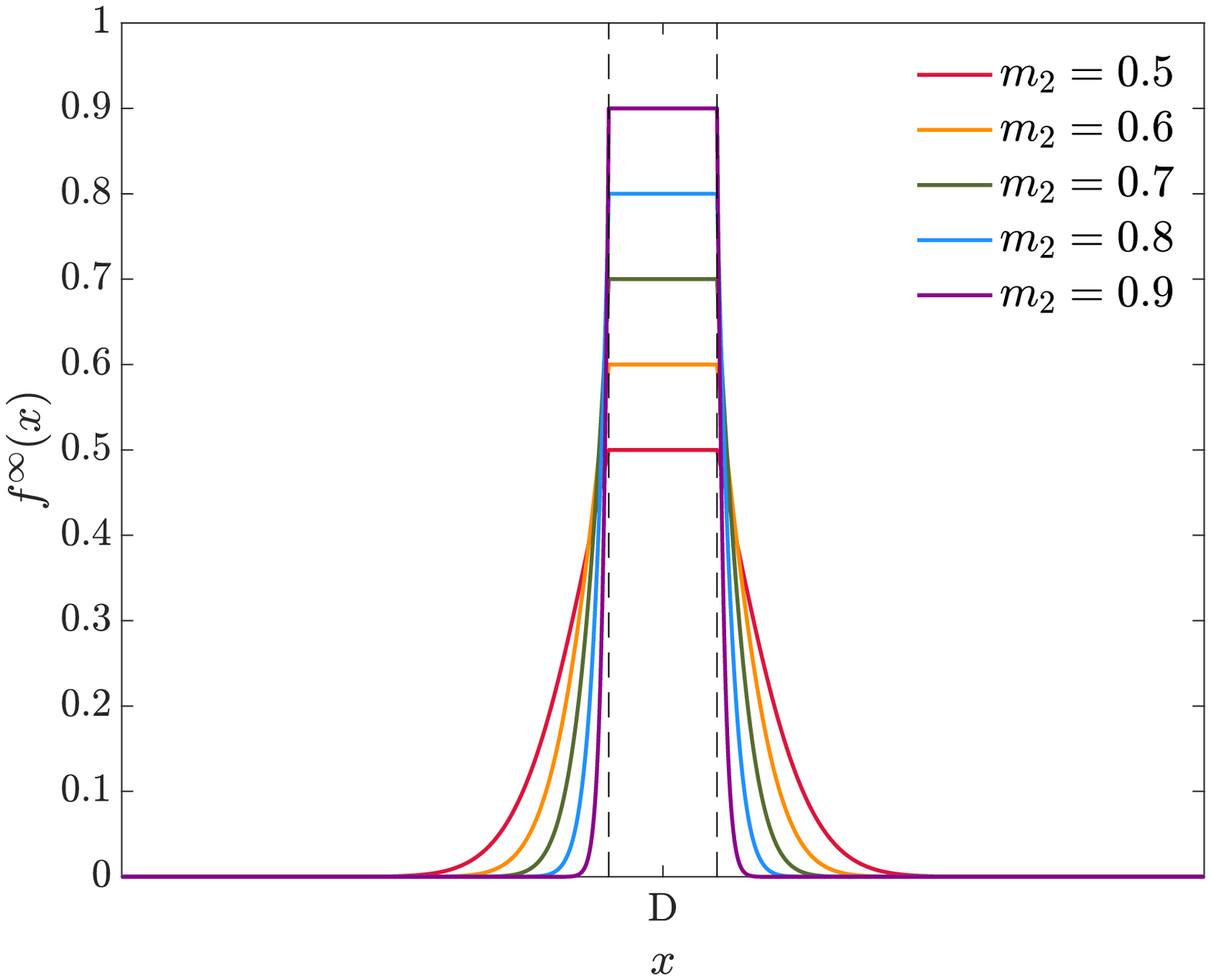}
\includegraphics[scale = 0.375]{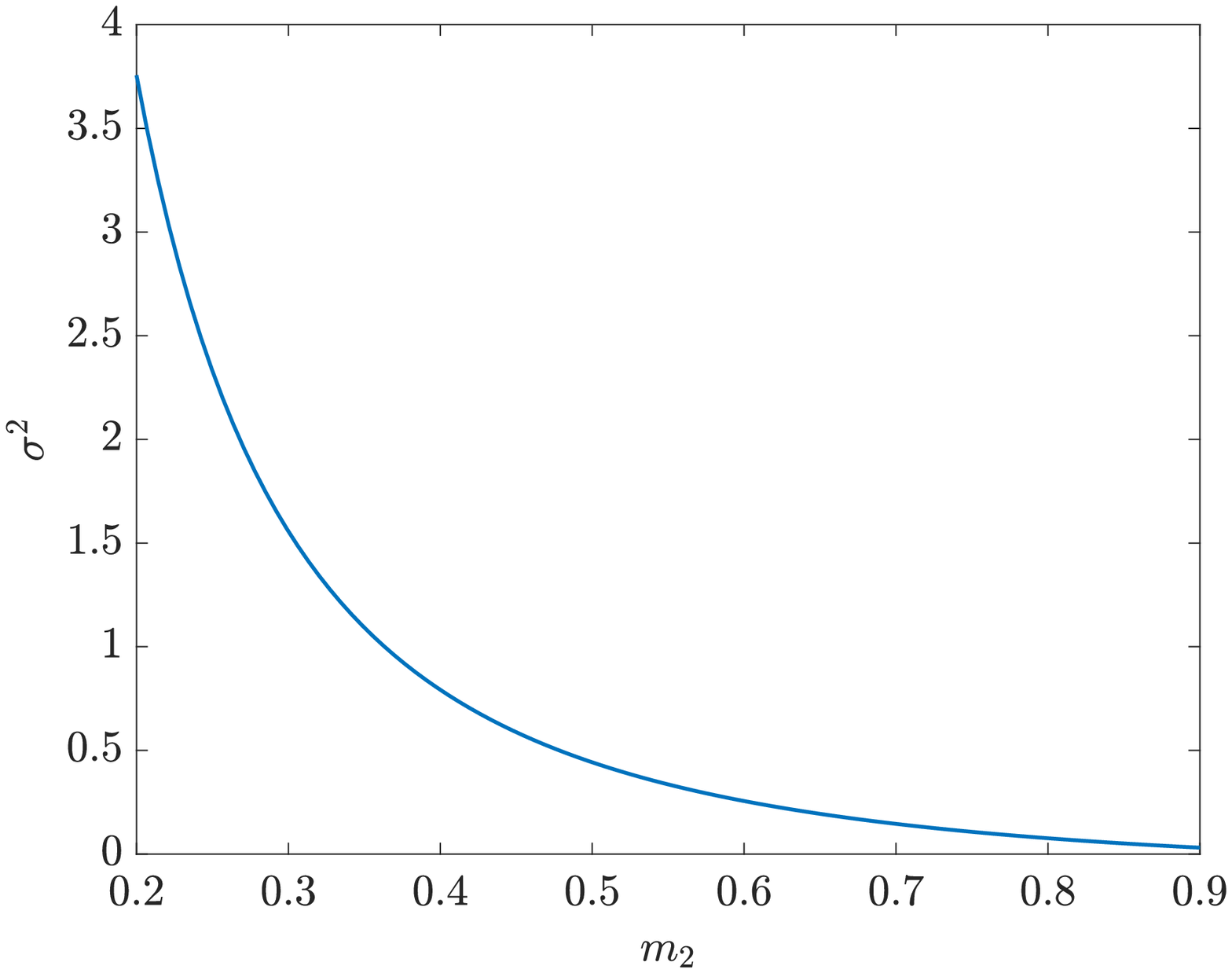}
\caption{Left: we represent the steady state \eqref{eq:finf} where $D = [-\delta,\delta]$, $x_0 = 0$ and $\delta = \frac{1}{2}$. We fixed the value of $m_2$ and we solved numerically \eqref{eq:sys} to determine $\sigma^2$.  Right: values of $\sigma^2$ determined from \eqref{eq:sys} for several values of $m_2 \in [0.2,0.9]$.  }
\label{fig:mass}
\end{figure}

\subsection{2D case}

Let us now consider the case $d = 2$. Proceeding as in Section \ref{sub:1D} we easily conclude that the equilibrium distribution is given by \fer{eq:equi2}, which for $d =2$ takes the form
\be
\label{eq:eq_mD}
f^\infty(\x) = 
\begin{cases}
\dfrac{m_1}{2\pi\sigma^2} \exp\left\{ -\dfrac{|\x-\x_0|^2}{2\sigma^2}\right\} & |\x-\x_0|\ge \delta, \\
\dfrac{m_2}{\delta^2 \pi} & |\x-\x_0| < \delta.
\end{cases}
\ee
Similarly to the 1D case, the values of $m_1,m_2>0$ are determined by imposing conservation of the total mass of the distribution and continuity at the boundaries of $D$. In details, by imposing $\int_{\mathbb R^2} f^\infty(\mathbf{x})dxdy = 1$  we get
\[
\dfrac{m_1}{2\pi\sigma^2} \int_{|\mathbf x - \mathbf x_0|>\delta} \exp\left\{ -\dfrac{|\mathbf{x}-\mathbf{x}_0|^2}{2\sigma^2}\right\}dxdy + m_2 = 1,
\]
that can be now easily solved in polar coordinates. Furthermore, the continuity at the interface is obtained by imposing 
\[
\lim_{|\mathbf{x}- \mathbf{x}_0| \rightarrow \delta} \dfrac{m_1}{2\pi\sigma^2} \exp\left\{ -\dfrac{|\x-\x_0|^2}{2\sigma^2}\right\}=  \dfrac{m_2}{\delta^2 \pi}. 
\]
Therefore, the constants $m_1,m_2$ are solution of the following system
\be
\label{eq:prob_2D}
\begin{cases}
m_1 \exp\left\{ -\dfrac{\delta^2}{2\sigma^2}\right\} + m_2 = 1 \\
\dfrac{m_1}{2\sigma^2}\exp\left\{ -\dfrac{\delta^2}{2\sigma^2}\right\} -\rev{ \dfrac{m_2}{\delta^2}} = 0.
\end{cases}
\ee
For given $\delta>0$ and $\sigma^2$ we obtain
\[
m_1 = e^{\delta^2/2\sigma^2} \dfrac{2\sigma^2}{2\sigma^2 + \delta^2}>0, \qquad m_2 = \rev{\dfrac{\delta^2}{2\sigma^2 + \delta^2}}>0. 
\]
As before, we can investigate whether, given $\delta>0$ and $m_2>0$, there exists a unique value for $\sigma^2$ and $m_1$ solution to \eqref{eq:prob_2D}. Letting $x = \dfrac{\delta}{\sqrt{2\sigma^2}}>0$ we may rewrite \eqref{eq:prob_2D}  as
\[
\begin{cases}
m_1 e^{-x^2} = 1-m_2 \\
m_1 e^{-x^2} -\rev{ \dfrac{m_2}{x^2 }} = 0,
\end{cases}
\]
from which, setting as before $m_* = 1-m_2$, we get $m_1 = m_* e^{x^2}$. Therefore, we have
\[
\rev{
x^2 = \dfrac{1-m_*}{ m_*}.}
\]
Hence, there exists a unique $\bar x \in (0,+\infty)$ solving the above equation, and $m_1,\sigma^2>0$ are uniquely determined as 
\[
\sigma^2 = \dfrac{\delta^2}{2\bar x^2}, \qquad m_1 = e^{\bar x^2} m_*. 
\]
Note that the computations in the 2D case are simpler than the computations of Section \ref{sub:1D}.
\begin{figure}
\centering
\includegraphics[scale = 0.35]{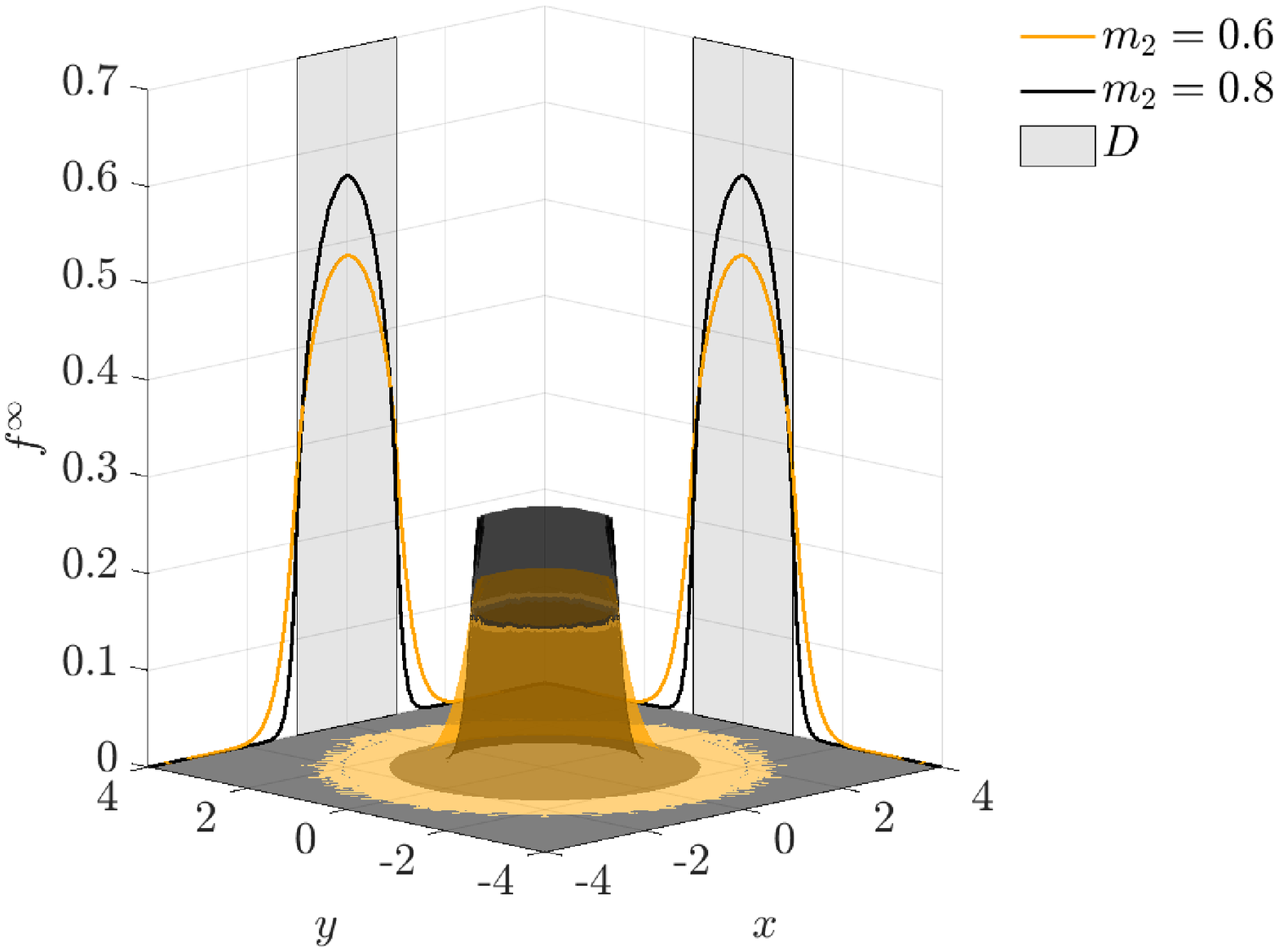}
\includegraphics[scale = 0.35]{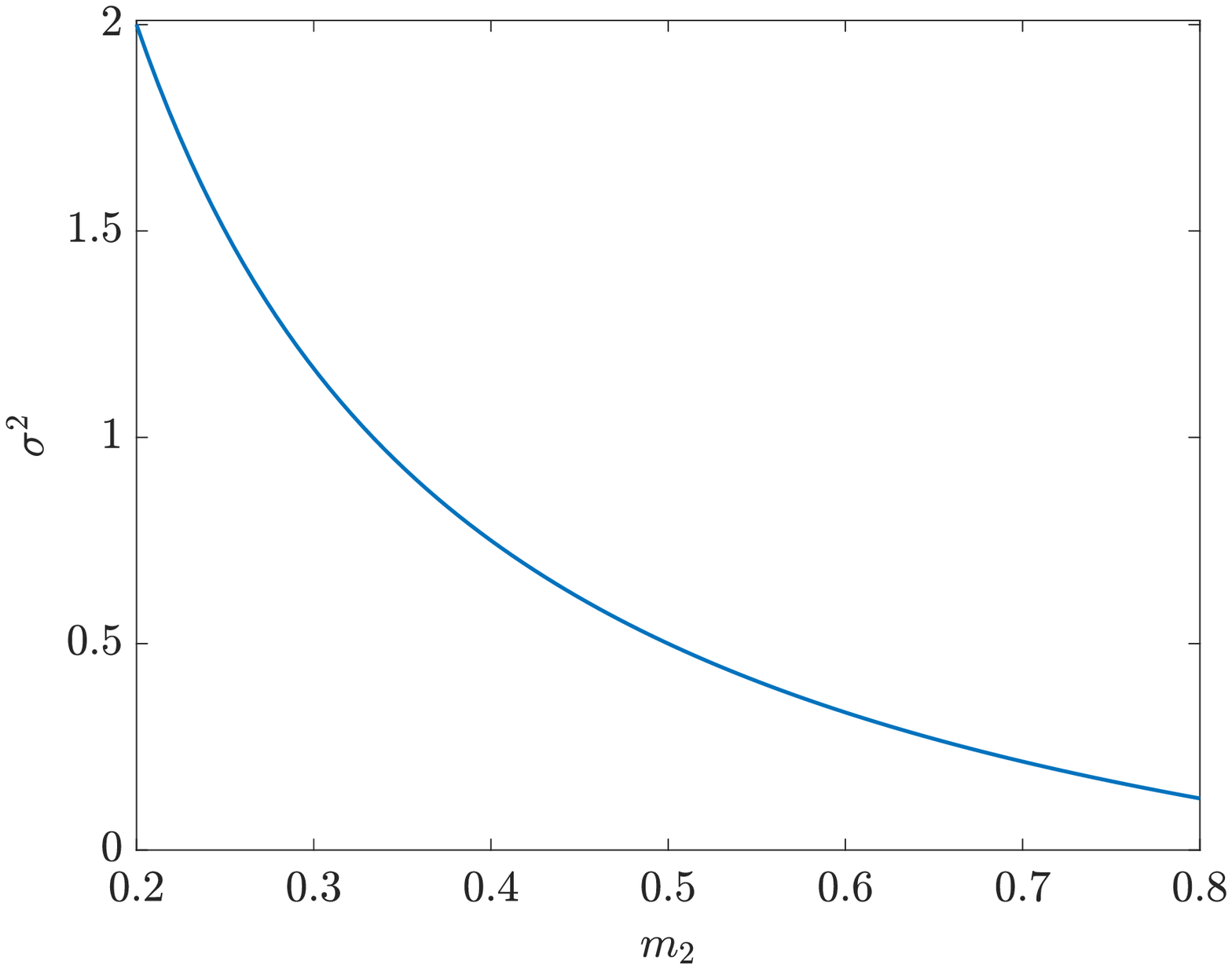}
\caption{Left: steady state defined by \eqref{eq:eq_mD} for two fixed values of \rev{$m_2 = 0.6,0.8$}, and for the domain $D = \{\mathbf x \in \mathbb R^2: |\mathbf x - \mathbf x_0|\le \delta\}$ with $\delta = 1$ and $\mathbf x_0 = (0,0)$. The values of $m_1,\sigma^2>0$ have been determined as the numerical solution to \eqref{eq:prob_2D}. We plot the marginal distributions $\int_{\mathbb R} f^\infty(\mathbf x)dx$ and $\int_{\mathbb R} f^\infty(\mathbf x)dy$ on the planes $xz$ and $yz$.  Right: values of $\sigma^2>0$ solution to \eqref{eq:prob_2D} for several values of \rev{$m_2 \in [0.2,0.8]$}. }
\label{fig:2D_system}
\end{figure}

In Figure \ref{fig:2D_system} we depict the steady state \eqref{eq:eq_mD} for two choices of  $m_2= 0.5,0.9$ where we considered the target domain $D = \{\mathbf x \in \mathbb R^2:|\mathbf x - \mathbf x_0| \le 1\}$ with $\mathbf x_0 = (0,0)$. Once fixed $m_2>0$ and $\delta=1$ we solved system \eqref{eq:prob_2D} numerically to get the unique values of $m_1>0$ and $\sigma^2>0$ that guarantee continuity at the boundaries of $D$ and conservation of the total mass. On the right plot we depict the resulting values of $\sigma^2$ for $m_2 \in [0.2,0.9]$. We may easily observe how, similarly to the 1D case, the value of $\sigma^2>0$ rapidly decays of large $m_2$.

\section{Trend to equilibrium}
\label{sec:trends}

In this section we restrict our analysis to the simplest case $d =1$. We aim to discuss the validity of the choice of resorting to the Fokker--Planck type equation \eqref{eq:model}, characterized by a constant coefficient of diffusion, in relation with the possibility to rigorously study the relaxation to equilibrium of its density function solution $f(x,t)$, $t >0$, $x \in \R$.  The Fokker--Planck equation  \eqref{eq:model} is characterized  by constant diffusion and by the confining potential
\[
P(x) = 
\begin{cases}\vspace{0.25cm}
\dfrac{\delta^2}{2} & |x-x_0|<\delta \\
\dfrac{(x-x_0)^2}{2} & |x-x_0|\ge\delta, 
\end{cases}
\] 
that is not differentiable in $x = x_0 \pm \delta$. The equilibrium density of a Fokker-Planck model is related to the confinement potential by the formula \cite{Risken}  
\[
f^\infty(x) = C e^{-P(x)}. 
\]
The rate of convergence of the solution of \eqref{eq:model} towards equilibrium  has been deeply studied for several classes of function $\psi(x,x_0)$. In particular, a rigorous proof of exponential convergence to equilibrium is restricted to strongly convex potentials, i.e. $d^2/dx^2 P(x)\ge c >0$ \cite{OV}. Hence, these results are not directly applicable to equation \fer{eq:model}. A possible strategy could rely in a suitable correction of the potential to obtain convexity, but this methods will destroy the possibility to have a steady state with a flat profile inside $D$, which is exactly the goal of the model. A second possibility is to resort to the Fokker--Planck equation \fer{eq:model2} which, while maintaining the same steady profile, is characterized by variable diffusion coefficient and a drift derived by a strongly convex potential. 

In 1D, the Fokker--Planck equation \fer{eq:model2} has the form
\be\label{eq:FP}
\partial_t f(x,t) = \partial_x  \left[ (x-x_0) f(x,t) + \partial_x (\kappa(x)f(x,t))\right],
\ee
where $\kappa(x)$ is a continuous non constant diffusion function of the form 
\be
\label{eq:kappa}
\kappa(x) = 
\begin{cases}
\sigma^2 + \dfrac{\delta^2}{2} - \dfrac{1}{2}(x-x_0)^2 & |x-x_0|< \delta \\
\sigma^2 & |x-x_0|\ge\delta.
\end{cases}
\ee
It is immediate to verify that equation \fer{eq:FP} has the same stationary solution \fer{eq:finf} of the original Fokker--Planck equation \fer{eq:model}.
We may observe that the diffusion coefficient $\kappa(x)$ is not differentiable at the points $x = x_0 \pm \delta$. 

To proceed, we introduce a $C^2(\R)$ regularization $\kappa_\epsilon(x)$, $\e  \ll 1$, of \eqref{eq:kappa} such that $\kappa_\e(x) = \kappa(x)$ in the domain $D_\e= |x-x_0| \le \delta -\e$, and $\kappa_\epsilon(x) \rightarrow \kappa(x)$ uniformly for $\epsilon\rightarrow 0^+$. More details on its possible form will be discussed in the next section. Thanks to the considered regularization, we consider the following surrogate model 
\be
\label{eq:model_epsilon}
\partial_t f_\epsilon(x,t) = \partial_x  \left[ (x-x_0) f_\epsilon(x,t) + \partial_x (\kappa_\epsilon(x)f_\epsilon(x,t))\right],
\ee
whose large time distribution of unit mass is now $f^\infty_\epsilon(x)$. We outline that the shape of the equilibrium is heavily dependent of the introduced regularization $\kappa_\epsilon(x)$, and it is generally not known. However, by construction we know that its profile is flat in the domain $D_\e\subset D$. The equilibrium density is solution to 
\begin{equation}
\label{eq:fepsilon_cond}
((x-x_0)  + \kappa_\epsilon^\prime(x))f_\epsilon(x,t) + \kappa_\epsilon(x) \partial_x f_\epsilon(x,t) = 0. 
\end{equation}
The following result holds. 

\begin{theorem}
Let $f_\epsilon(x,t)$ be the solution to \eqref{eq:model_epsilon} and let $\|\kappa_\epsilon(x) \|_{L^\infty}<C$. If $f_\epsilon(x,0) \in L^2(\R)$ then  $f_\epsilon \in L^2(\R \times [0,+\infty))$ for any $\epsilon>0$ and
\[
\|f_\epsilon(t) \|_{L^2}^2 \le e^{2t} \| f(0)\|_{L^2}
\]
\end{theorem}
\begin{proof}
We multiply \eqref{eq:model_epsilon} by $2f_\epsilon(x,t)$ and we integrate over $\R$
\[
\begin{split}
\dfrac{d}{dt} \| f_\epsilon(t)\|^2_{L^2} &= \int_\R 2f_\epsilon \partial_x\left[(x-x_0)f_\epsilon + \partial_x (\kappa_\epsilon(x)f_\epsilon)  \right]dx \\
&=  \int_\R 2 f_\epsilon \left[f_\epsilon + (x-x_0)\partial_x f_\epsilon + \partial_x^2(\kappa_\epsilon(x)f_\epsilon) \right]dx \\
&= 2 \| f_\epsilon\|_{L^2}^2 + 2\int_\R f_\epsilon (x-x_0)\partial_x f_\epsilon dx+ 2 \int_\R f_\epsilon \partial_x^2 (\kappa_\epsilon(x)f_\epsilon )dx
\end{split}
\]
Noticing that 
\[
\begin{split}
2\int_\R f_\epsilon (x-x_0)\partial_x f_\epsilon dx = - \int_\R f_\epsilon^2dx  = -\| f_\epsilon\|^2_{L^2},
\end{split}
\]
and that 
\[
\int_\R f_\epsilon \partial_x^2 (\kappa_\epsilon(x)f_\epsilon )dx \le -C \|\partial_x f_\epsilon \|_{L^2}^2,
\]
we have
\[
\dfrac{d}{dt} \|f_\epsilon \|_{L^2}^2 \le 2 \|f_\epsilon(t) \|^2_{L^2}, 
\]
and thanks to the Gronwall inequality we conclude. 
\end{proof}

For the classical Fokker-Planck equation, where $\kappa_\epsilon(x) \equiv 1$, the steady state of unit mass is the Gaussian density 
\[
g(x) = \dfrac{1}{\sqrt{2\pi\sigma^2}} \exp\left\{ - \dfrac{(x-x_0)^2}{2}\right\}.
\]
In this case, if the relative Shannon entropy between the solution $f_\epsilon(\cdot,t)$ and the equilibrium density $g$, given by 
\[
H(f(t)|g) = \int_{\R} f(x,t)\log\dfrac{f(x,t)}{g(x,t)}dx, 
\]
is bounded at $t = 0$, then this quantity decays exponentially in time towards zero ensuring convergence towards equilibrium in $L^1(\R)$ with an explicit exponential rate \cite{OV}. Few results about convergence to equilibrium are available in the case of non constant diffusion coefficients. These results make use of new differential inequalities, like Chernoff inequality \cite{FPTT}, and generally exponential convergence to equilibrium is lost.  

\rev{To enlighten the understanding of the mathematical methods which allow to study the convergence in time of the solution to Fokker--Planck type equations like \fer{eq:model_epsilon}, let us briefly summarize the analysis of \cite{FPTT, TT}.}

In detail, let us rewrite equation \eqref{eq:model_epsilon} for the quotient $F_\epsilon(x,t) = f_\epsilon(x,t)/f^\infty_\epsilon(x)$
\begin{equation}
\label{eq:model2_epsilon}
\partial_t F_\epsilon(x,t) = \kappa_\epsilon(x) \partial_x^2 F_\epsilon(x,t) - (x-x_0)\partial_x F_\epsilon(x,t). 
\end{equation}

We recall the following result.
\begin{theorem}\cite{FPTT}\label{th:3}
Let us consider the smooth convex function $\Phi(x)$, $x \in \R_+$. Then if $F_\epsilon(x,t)$ is the solution of \eqref{eq:model2_epsilon} and $c \le F_\epsilon(x,t)\le C$, for $0<c<C$ and any $\epsilon>0$, the functional 
\[
\Theta(F(t)) = \int_\R f^\infty_\epsilon(x) \Phi(F_\epsilon(x,t))dx
\]
is monotonically decreasing in time and 
\[
\dfrac{d}{dt} \Theta(F_\epsilon(t))= -I_\Theta(F_\epsilon(t)), 
\]
where the entropy production term $I_\Theta$ is given by
\[
I_\Theta(F_\epsilon(t)) = \int_\R \kappa_\epsilon(x)f_\epsilon^\infty(x)\Phi^{\prime\prime}(F_\epsilon(x,t))\left| \partial_x F_\epsilon(x,t)\right|^2dx\ge 0.
\]
\end{theorem}

Since the relative Shannon entropy is obtained by choosing $\Phi(x) = x\log(x)$, Theorem \ref{th:3} implies that the entropy production can be written as
\[
\begin{split}
I_S(F_\epsilon(t)) &= \int_\R \kappa_\epsilon(x) f^\infty_\epsilon(x) \dfrac{1}{F_\epsilon(x,t)}|\partial_x F_\epsilon(x,t)|^2dx \\
&= 4\int_\R \kappa_\epsilon(x) f^\infty_\epsilon(x,t)\left( \partial_x \sqrt{\dfrac{f_\epsilon(x,t)}{f^\infty_\epsilon(x)}}  \right)^2dx
\end{split}
\]
If now we consider the convex function $\Phi(x) = (\sqrt{x}-1)^2$,  the functional $\Theta(F_\epsilon(x,t))$ coincides with 
\[
\int_{\R} f^\infty_\epsilon(x) \left( \sqrt{F_\epsilon(x,t)}-1\right)^2dx = \int_\R \left( \sqrt{f_\epsilon(x,t)} - \sqrt{f^\infty_\epsilon(x)}\right)^2dx, 
\]
namely with the square root of the Hellinger distance 
\be
\label{eq:dH}
d_H(f,g) =\left( \int_\R \left(\sqrt{f(x)} - \sqrt{g(x)}\right)^2\right)^{1/2}.
\ee
Hence, applying Theorem \ref{th:3} to $\Phi(x) = (\sqrt{x}-1)^2$,  we conclude that the Hellinger distance is monotonically decreasing \cite{FPTT}
\be
\label{eq:dH_mon}
\dfrac{d}{dt} d_H(f_\epsilon,f^\infty_\epsilon) \le 0. 
\ee
Following \cite{FPTT}, we now resort to an inequality relating the differential expression of the stationary solution of the Fokker--Planck equation \fer{eq:model2}.

\begin{theorem}[Chernoff with weight \cite{FPTT}]
Let $X$ be a random variable with density $f^\infty_\epsilon(x)$, $x \in \R$, where $f^\infty_\epsilon $ is solution to 
\[
\partial_x (\kappa_\epsilon f^\infty_\epsilon(x)) + (x-x_0)f^\infty_\epsilon(x) = 0, \qquad x \in \R. 
\]
For any given  absolutely continuous function $\psi(x)$, with in $\R$ and $\textrm{Var}[\psi(X)]<+\infty$
\[
\textrm{Var}[\psi(X)] \le \mathbb E[\kappa(X)\left(\psi^\prime(X)\right)^2]
\]
\end{theorem}

Applying Chernoff inequality with $\psi(\cdot) = \sqrt{f_\epsilon(\cdot,t)/f_\epsilon^\infty(\cdot)}$ we thus obtain
\[
  4 \left(1-\left(\int_\R \sqrt{f_\epsilon(x,t) f^\infty_\epsilon(x)}dx \right)^2\right) \le 4\int_\R \kappa_\epsilon(x) f_\epsilon^\infty(x) |\partial_x \psi(x)|^2 dx = I_S(f_\epsilon,f^\infty_\epsilon). 
\]
We can now relate the left-hand side of the inequality above  with the Hellinger distance by resorting to the following argument.  For any pair $f,g$ of probability density functions we have the inequality
\[
\begin{split}
d_H^2(f,g) &= \int_\R \left(f(x) + g(x)-2\sqrt{f(x)g(x)}dx \right)dx \\
&= 2 \left(1- \int_\R \sqrt{f(x)g(x)}dx \right) \le 2 \left(1- \left(\int_\R \sqrt{f(x)g(x)}dx \right)^2\right). 
\end{split}\]
Indeed, by Cauchy-Schwarz inequality we have
\[
\int_\R \sqrt{f(x)} \sqrt{g(x)}dx \le 1. 
\]
Finally, we have the bound
\[
 2d_H^2(f_\epsilon,f^\infty_\epsilon) \le I_S(f_\epsilon,f^\infty_\epsilon),
\]
and the decay of the Shannon relative entropy coupled with the previous bound implies 
\[
\dfrac{d}{dt} H(f_\epsilon,f^\infty_\epsilon) = -I_S(f_\epsilon,f^\infty_\epsilon) \le -2d_H^2(f_\epsilon,f^\infty_\epsilon). 
\]
Integrating with respect to time from $0$ to $\infty$ we get
\[
2 \int_0^\infty d^2_H(f_\epsilon,f^\infty_\epsilon)dx \le H(f_\epsilon(0),f^\infty_\epsilon),
\]
which shows that $d^2_H \in L^1(\R_+)$. Furthermore, since $d_H$ is monotone in time, see \eqref{eq:dH_mon}, we have shown that
\[
d_H^2(f_\epsilon,f^\infty_\epsilon) = o\left( t^{-1} \right),\quad t\rightarrow +\infty.
\]
In particular, since the previous argument holds for any value  $\epsilon>0$, the solution to \eqref{eq:model_epsilon} converges in time towards its steady state $f_\epsilon^\infty(x)$, given by \eqref{eq:fepsilon_cond}, and  the rate of convergence is at least linear in time. \rev{ We highlight how, as shown in \cite{TT} to which the interested reader is referred for details, we can always lift both the initial value and the equilibrium solution of the Fokker-Planck type equation \fer{eq:model_epsilon} to rigorously apply the previous strategy, and subsequently remove the lifting without loosing the convergence rate}.


\section{Numerical results}
\label{sec:num}

In this section we investigate numerically the trends to equilibrium of the introduced  models. We remark that the model with nonconstant diffusion \eqref{eq:FP}, with the diffusion function $\kappa(x)$ defined in \eqref{eq:k}, is equivalent to the initial model \eqref{eq:model} with discontinuous drift  \eqref{eq:drift}.  In particular, we will concentrate on the regularized problem embedding a precise form of $\kappa_\epsilon(x)$ for which we have no information on the analytical form of the large time behavior for each $\epsilon>0$. Hence, we apply a class of  recently developed structure preserving (SP)  schemes for Fokker-Planck-type equations, see \cite{PZ}. These methods are capable to reproduce large times statistical properties of the exact steady state with arbitrary accuracy, together with the preservation of positivity of the solution and a consistent entropy dissipation. This class of schemes has been extended to models with nonconstant diffusion matrices in \cite{LZ} and to preserve positivity of  stochastic Galerkin reformulations of linear problems in \cite{Z}. In order to be self-consistent, we summarize the main features of SP methods in 1D. The extension to higher dimensional problems will be treated later. 

Let us rewrite the original Fokker-Planck model \eqref{eq:FP} in flux form 
\[
\partial_t f(x,t) =  \partial_{x}\mathcal F[f](x,t),
\]
where
\be
\label{eq:an_flux}
\mathcal F[f](x,t) = \mathcal C(x)f(x,t) + \kappa_\epsilon(x) \partial_x f(x,t). 
\ee
is the flux, and we introduced the drift function $\mathcal C(x) = (x-x_0) + \partial_x \kappa_\epsilon(x)$.  Hence, introducing a uniform discretization of the domain $\{x_i\}_{i=1}^N$ with $\Delta x = x_{i+1}-x_i>0$ constant, and denoting $x_{i+1/2} = x_i + \Delta x/2$, we define the following conservative discretization 
\[
\dfrac{d}{dt} f_i(t) = \dfrac{\mathcal F_{i+1/2}[f](t) - \mathcal F_{i-1/2}[f](t)}{\Delta x}, \qquad i = 1,\dots,N.
\]
In particular, choose the numerical flux of the form 
\be
\label{eq:num_flux}
\mathcal F_{i+1/2}[f] = \tilde{\mathcal C}_{i+1/2}\tilde{f}_{i+1/2} + \kappa_{\epsilon}(x_{i+1/2})\dfrac{f_{i+1}-f_{i}}{\Delta x},
\ee
where $\tilde{f}_{i+1/2} $ is a convex combination of the values of $f$ in two adjacent cells $i$ and $i+1$, i.e.
\[
\tilde{f}_{i+1/2}  = (1-\delta_{i+1/2})f_{i+1} + \delta_{i+1/2}f_i,
\]
being $\delta_{i+1/2}$ suitable nonlinear weights and $\tilde{\mathcal C}_{i+1/2}$ a numerical drift that are obtained in such a way the method yields nonnegative solutions and preserve the steady state of the problem with arbitrary accuracy. Setting, in particular
\be
\label{eq:Ctilde}
\tilde{\mathcal C}_{i+1/2} = \dfrac{\kappa_{\epsilon}(x_{i+1/2})}{\Delta x} \int_{x_i}^{x_{i+1}} \dfrac{(x-x_0) + \partial_x \kappa_\epsilon(x)}{\kappa_\epsilon(x)}dx,
\ee
since the numerical flux vanishes, we obtain equating the ratios $f_{i+1}/f_i$ and $f(x_{i+1},t)/f(x_i,t)$ of the numerical and exact fluxes the following definition of nonlinear weights
\be
\label{eq:delta}
\delta_{i+1/2} = \dfrac{1}{\lambda_{i+1/2}} + \dfrac{1}{1-\exp(\lambda_{i+1/2})}, 
\ee
with $\lambda_{i+1/2} = \dfrac{\Delta x\,\tilde{\mathcal C}_{i+1/2}}{\kappa_{\epsilon,i+1/2}}$. In details, the following result holds

\begin{proposition}
The numerical flux \eqref{eq:num_flux} with $\tilde{\mathcal C}_{i+1/2}$, $\delta_{i+1/2}$ defined in \eqref{eq:Ctilde}-\eqref{eq:delta} vanishes when the analytical flux \eqref{eq:an_flux} is zero in the cell $[x_i,x_{i+1}]$. Moreover $\delta_{i+1/2} \in[0,1]$ for any $i=1,\dots,N$.
\end{proposition}
For a detailed proof we point the interested reader to \cite{PZ}. The introduced method admits an equivalent entropic formulation based on the entropy dissipation principle.  In this case, the numerical flux should be defined as
\[
\mathcal F_{i+1/2}^E = \kappa_{\epsilon,i+1/2} \tilde{f}^E_{i+1/2}\left( \dfrac{\tilde{\mathcal C}_{i+1/2}}{\kappa_{\epsilon,i+1/2}}+\dfrac{\log f_{i+1}-\log f_i}{\Delta x}\right),
\]
where
\[
\tilde{f}^E_{i+1/2} = 
\begin{cases}
\dfrac{f_{i+1}-f_i}{\log f_{i+1}-\log f_i} & f_{i+1} \ne f_i \\
f_{i+1} & f_{i+1} = f_i, 
\end{cases}
\]
being $\tilde{\mathcal C}_{i+1/2}$ defined as in \eqref{eq:Ctilde}.
\subsection{Convergence towards equilibrium}

In this section, we provide numerical evidence of the convergence of $f_\epsilon$, solution to the  model \eqref{eq:model_epsilon}, towards the steady state \eqref{eq:finf} corresponding to the equilibrium solution to \eqref{eq:model}--\fer{eq:model2}. We use  the following regularization $\kappa_\epsilon(x)$, $\epsilon >0$, of the function $\kappa(x)$
\be\label{eq:kepsi}
\kappa_\epsilon(x) = 
\begin{cases}
\sigma^2 & |x-x_0| >  \delta+ \epsilon \\
p(x)&  \delta-\epsilon < |x-x_0| <  \delta+\epsilon \\
\sigma^2+ \frac{\delta^2}{2} - \frac{1}{2}|x-x_0|^2 & |x-x_0|<\delta-\epsilon\\
p(x)&  -\delta -\epsilon < |x-x_0| <  -\delta +\epsilon 
\end{cases}
\ee
being $p(x) = a |x-x_0|^3 + b |x-x_0|^2 + c |x-x_0| + d$, and  $a,b,c,d \in \R$ are  solution of the linear system
\[
\begin{cases}
p(\delta + \epsilon ) = \sigma^2 \\
p(\delta -\epsilon) = \sigma^2 + \frac{ \delta^2}{2} - \frac{(\epsilon - \delta)^2}{2} \\
p^\prime(\delta + \epsilon ) = 0 \\
p^\prime(\delta - \epsilon ) = \epsilon - \delta,
\end{cases}
\]
that follows by imposing continuity of $p$ and of its derivatives at the boundaries.  


In Figure \ref{fig:moll} we represent the obtained $\mathcal C^1(\R)$ regularization of the local diffusion function $\kappa(x)$ provided by $\kappa_\epsilon(x)$ for several values of $\epsilon>0$.  

\begin{figure}
\centering
\includegraphics[scale = 0.37]{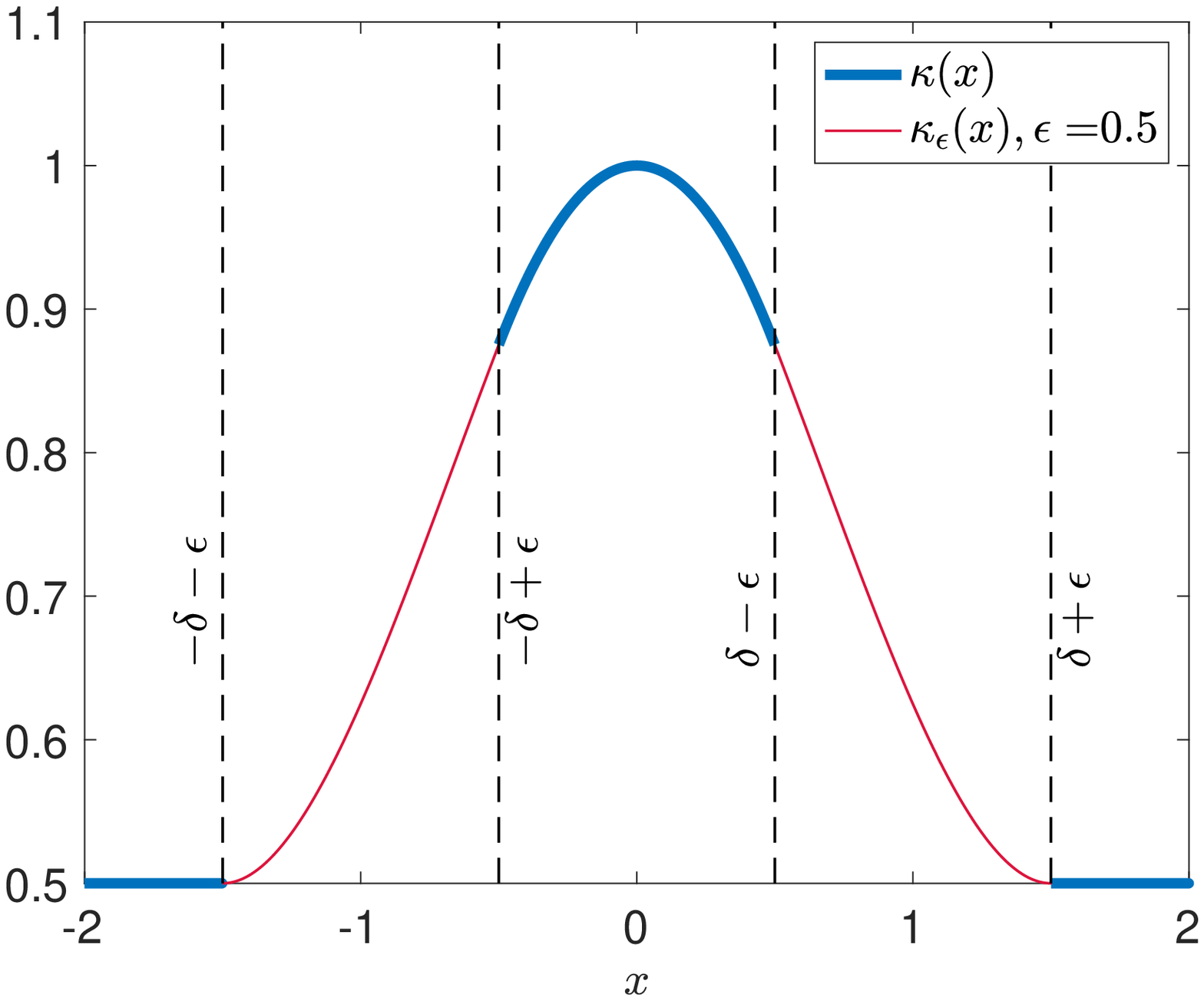}
\includegraphics[scale = 0.37]{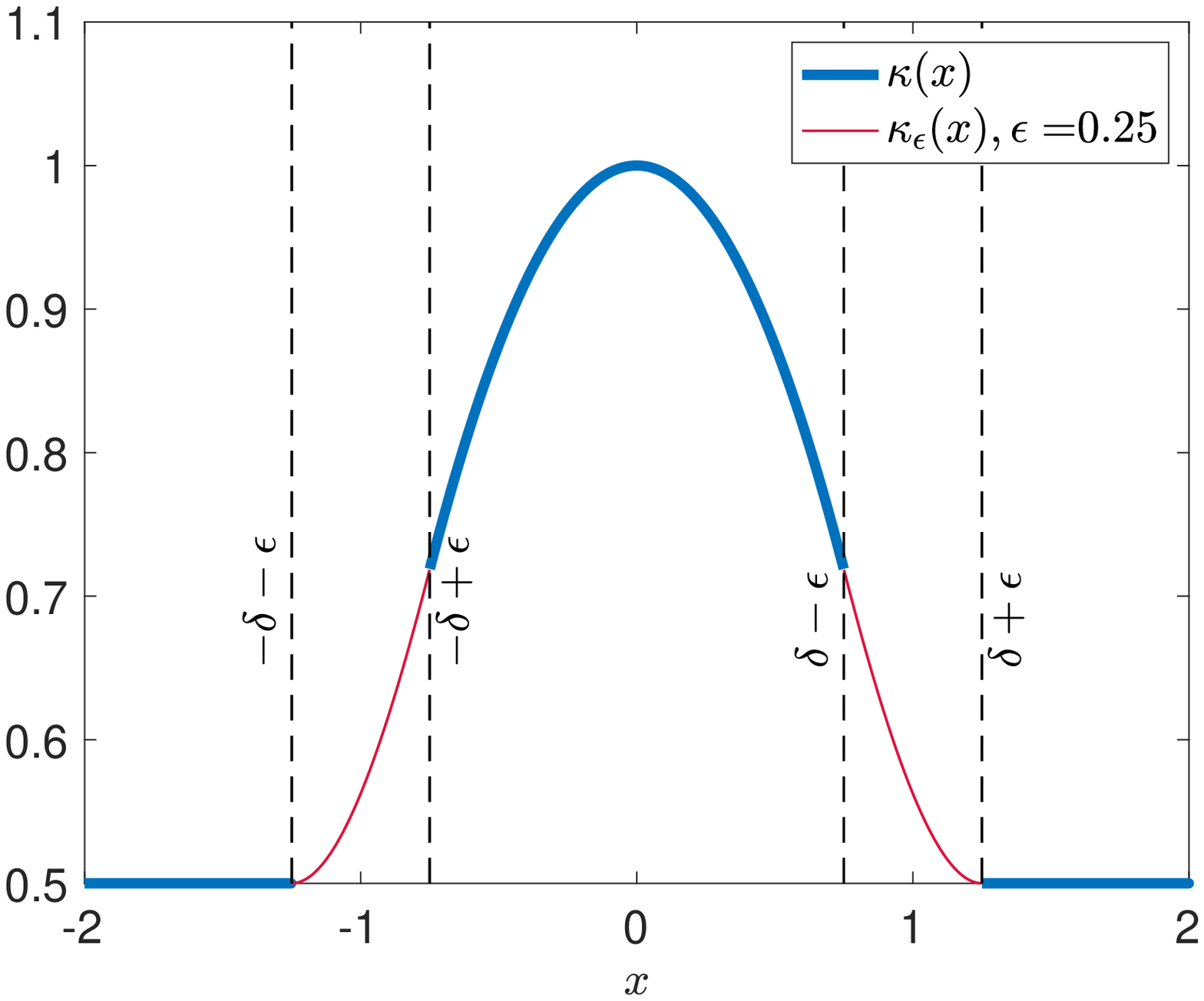} \\
\includegraphics[scale = 0.37]{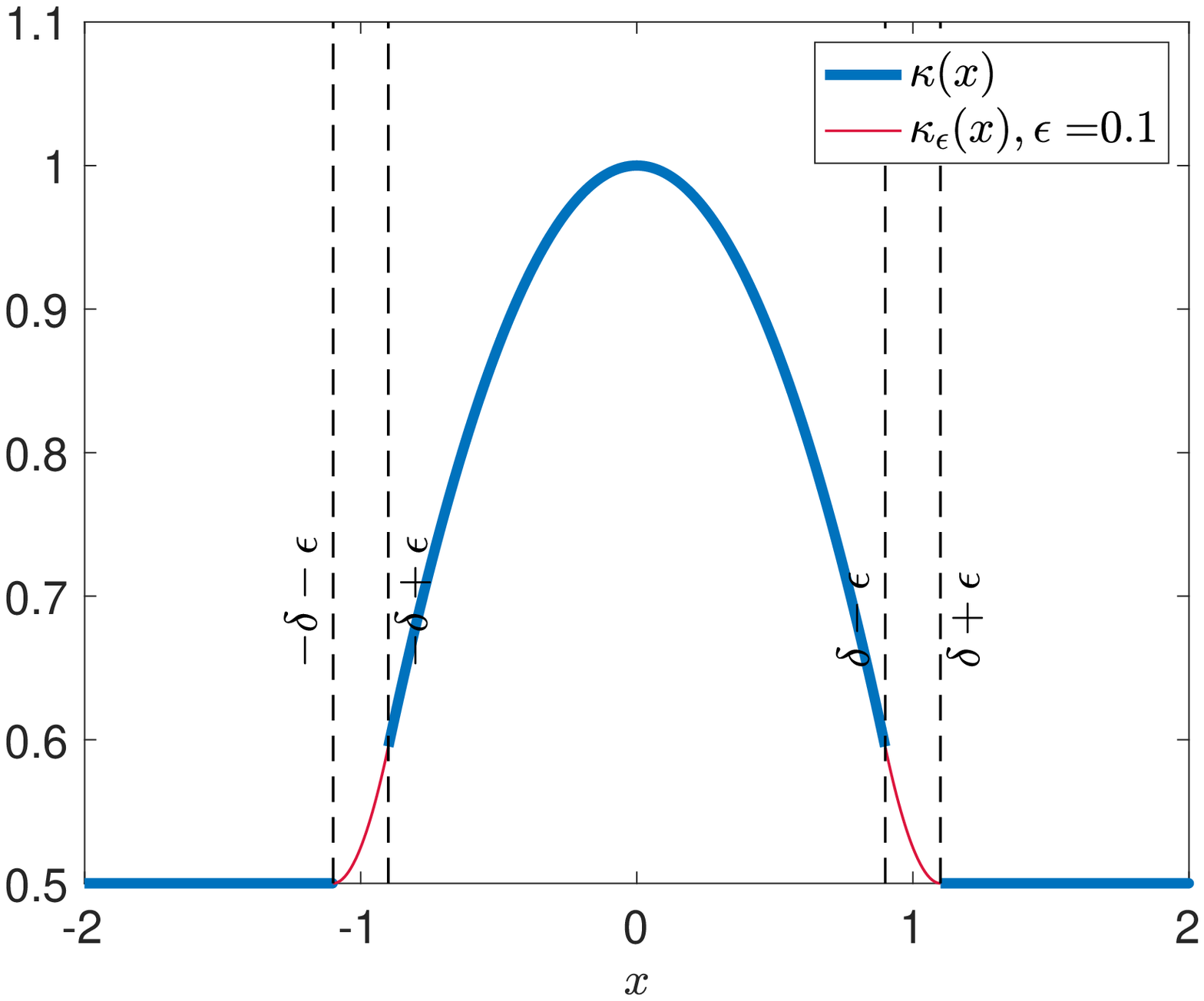}
\includegraphics[scale = 0.37]{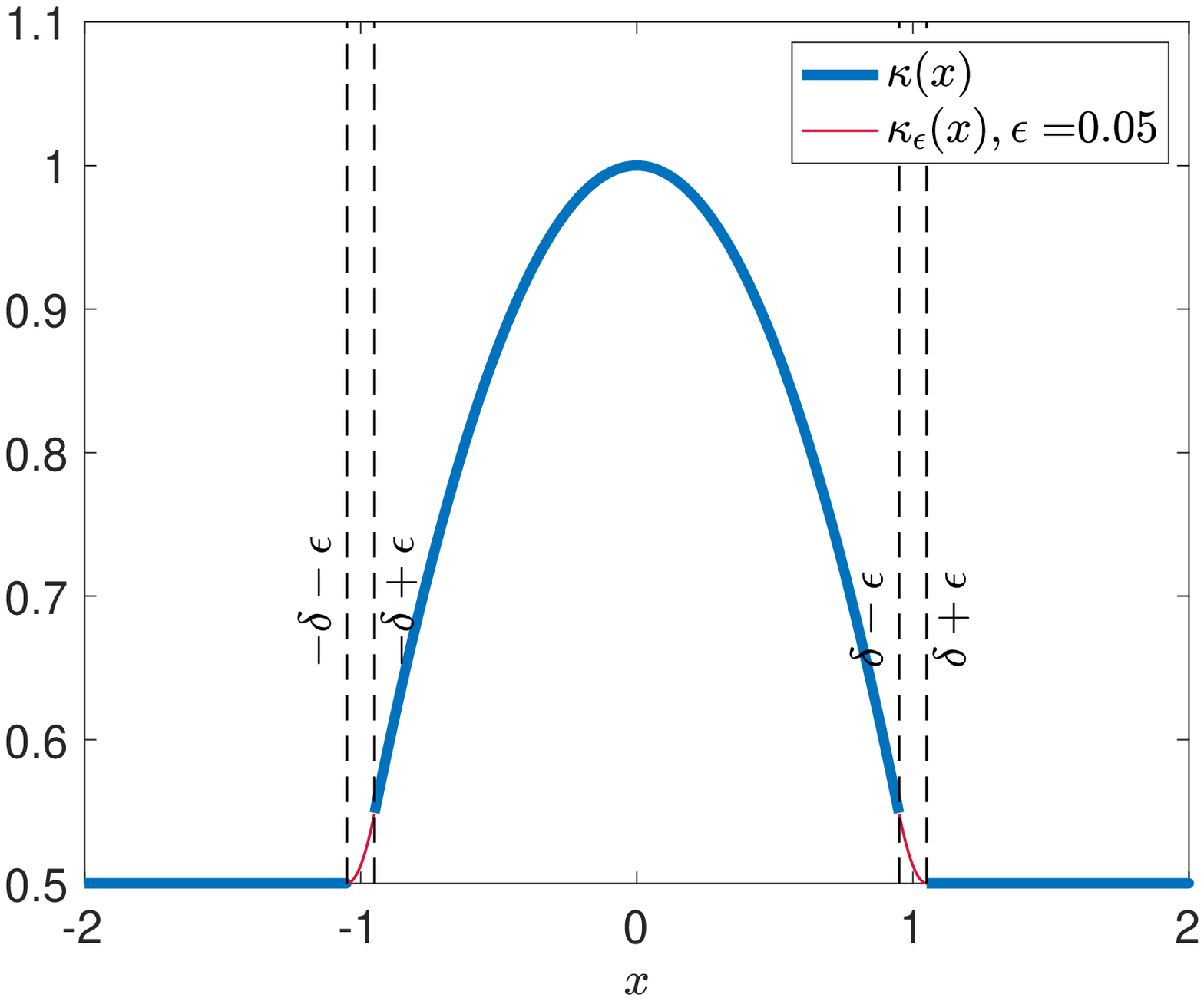}
\caption{Regularization $\kappa_\epsilon(x)\in \mathcal C^1(\R)$ of the local diffusion function $\kappa(x)$ defined by \eqref{eq:kepsi}. We considered $x_0=0$ and $\delta = 1$, $\sigma^2=1$ and several values of $\epsilon>0$.  }
\label{fig:moll}
\end{figure}

Once defined the regularized diffusion coefficient $\kappa_\epsilon(x)$ we focus on the approximation of the equilibrium distribution \eqref{eq:finf} by means of the surrogate model \eqref{eq:model_epsilon}. Since in this case the large time solution is unknown we apply structure preserving numerical methods for Fokker-Planck-type equations, see \cite{PZ}. We consider as initial distribution 
\begin{equation}
\label{eq:f0_num}
f(x,0) = \beta \left[\exp(-c(x+2)) + \exp(-c(x-2)) \right], \qquad c = 10,
\end{equation}
with $\beta>0$ a given parameter such that $\int_{\mathbb R}f(x,0)dx = 1$. Hence, we introduce a discretization of the domain $[-L,L]$, $L = 5$, obtained with \rev{$N = 81$} gridpoints such that $\Delta x = 2L/(N-1)$, and a uniform time discretization with $\Delta t = \Delta x^2/L^2$. The time integration has been performed with a standard RK4 method. In Figure \ref{fig:SP} we represent the obtained numerical approximation of $f_\epsilon^\infty(x)$ for several $\epsilon>0$ and the evolution of the relative error
\[
\textrm{Error}_\epsilon(t) = \int_{\R}\dfrac{|f_\epsilon(x,t)-f^\infty(x)|}{|f^\infty(x)|}dx, \qquad t\in [0,T].
\]
We may easily observe the consistency of the approximation of the large time solution in terms of $\epsilon>0$. In particular, for small $\epsilon>0$, the surrogate model with $\kappa_\epsilon(x) \in \mathcal C^1(\R)$ is capable to correctly approximate the analytical steady state \eqref{eq:finf}. 

\begin{figure}
\centering
\includegraphics[scale= 0.35]{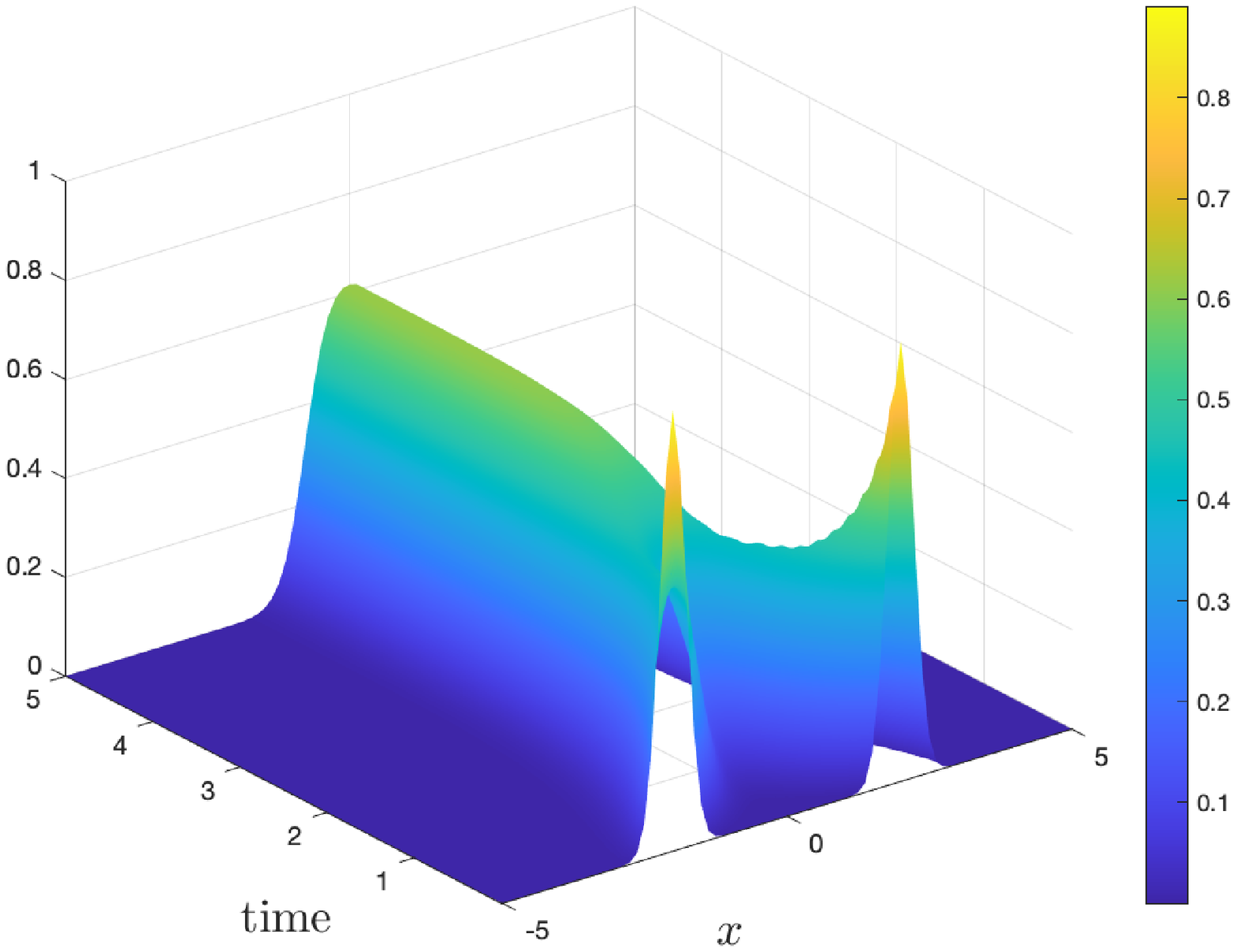}
\includegraphics[scale= 0.35]{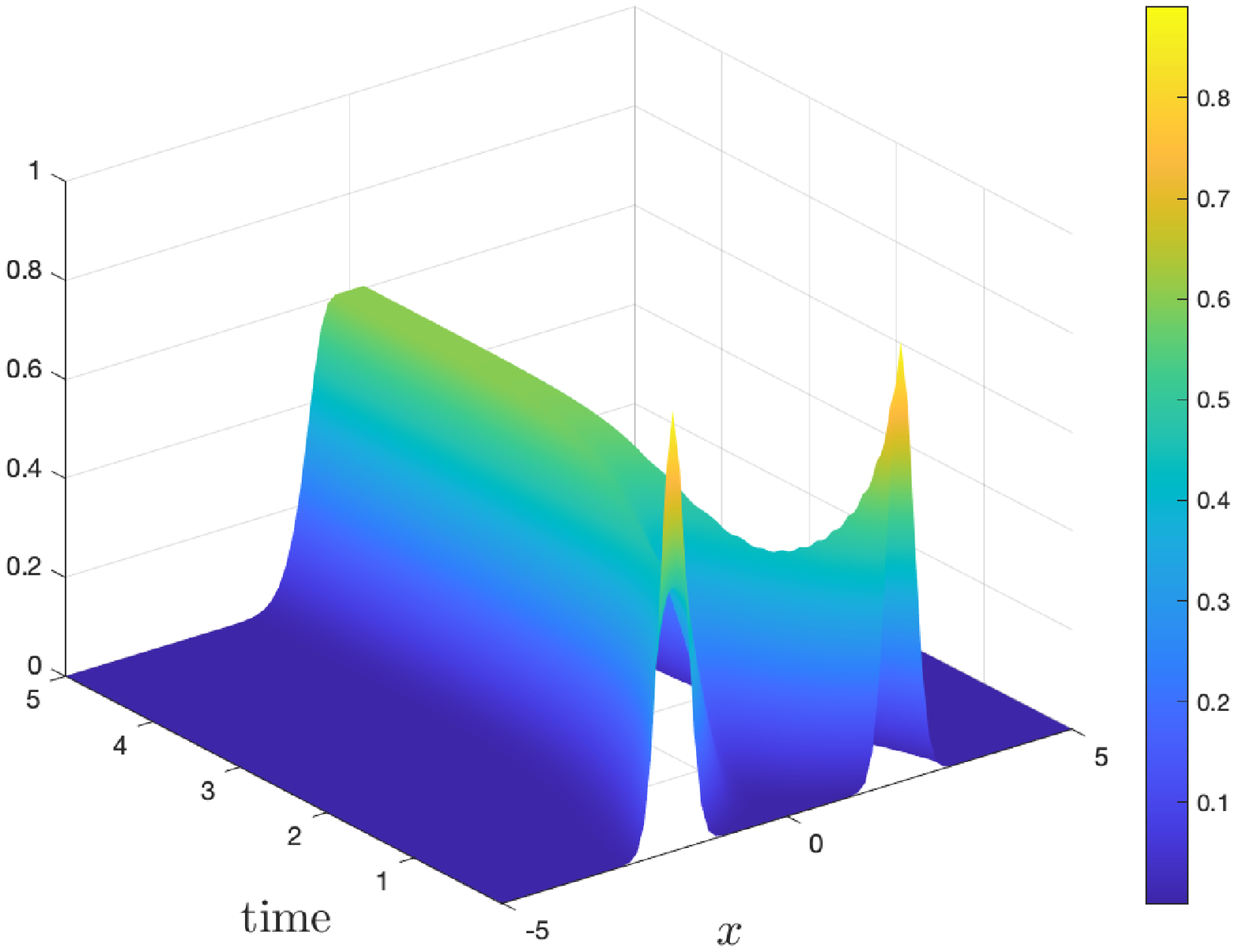} \\
\includegraphics[scale= 0.35]{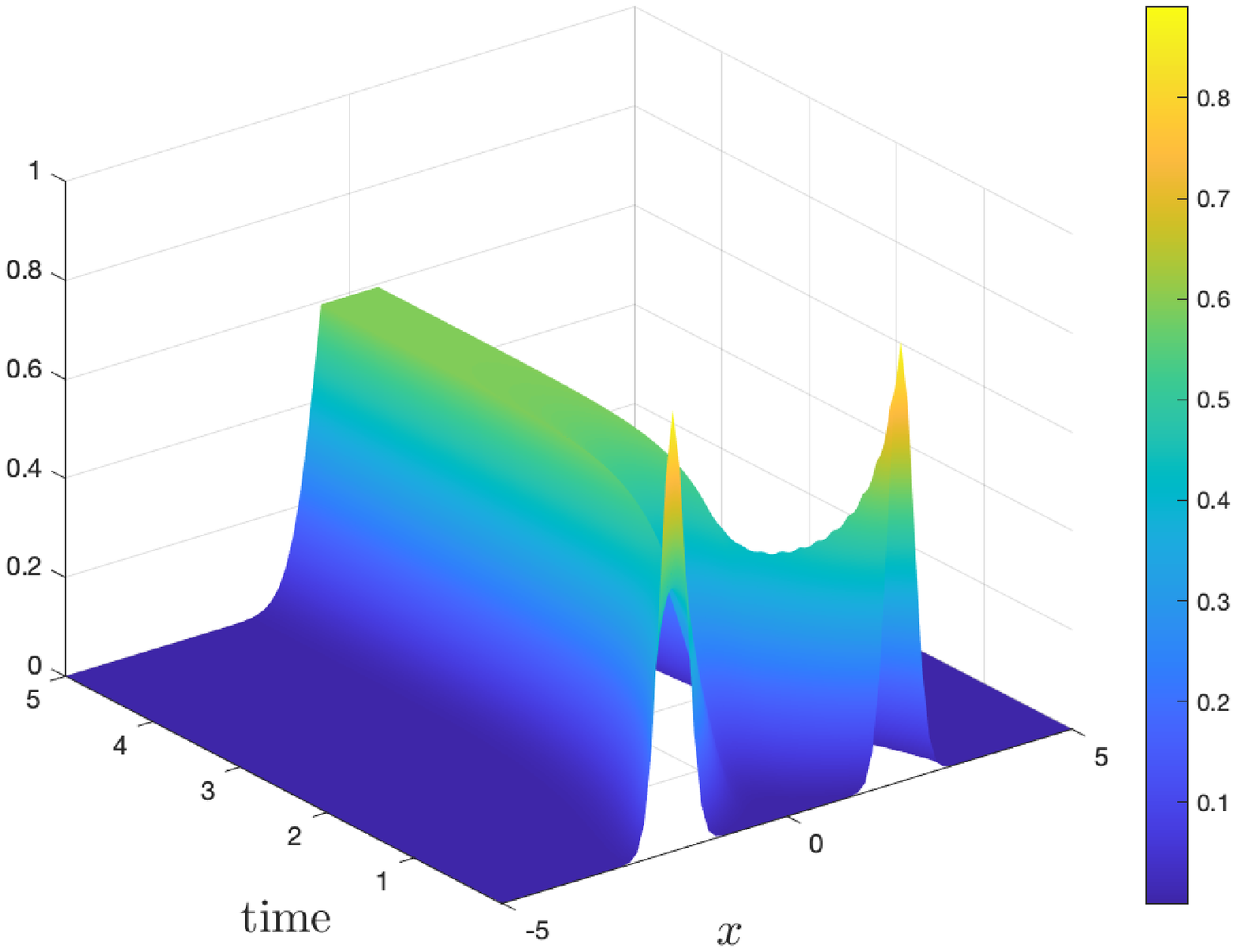} 
\includegraphics[scale = 0.35]{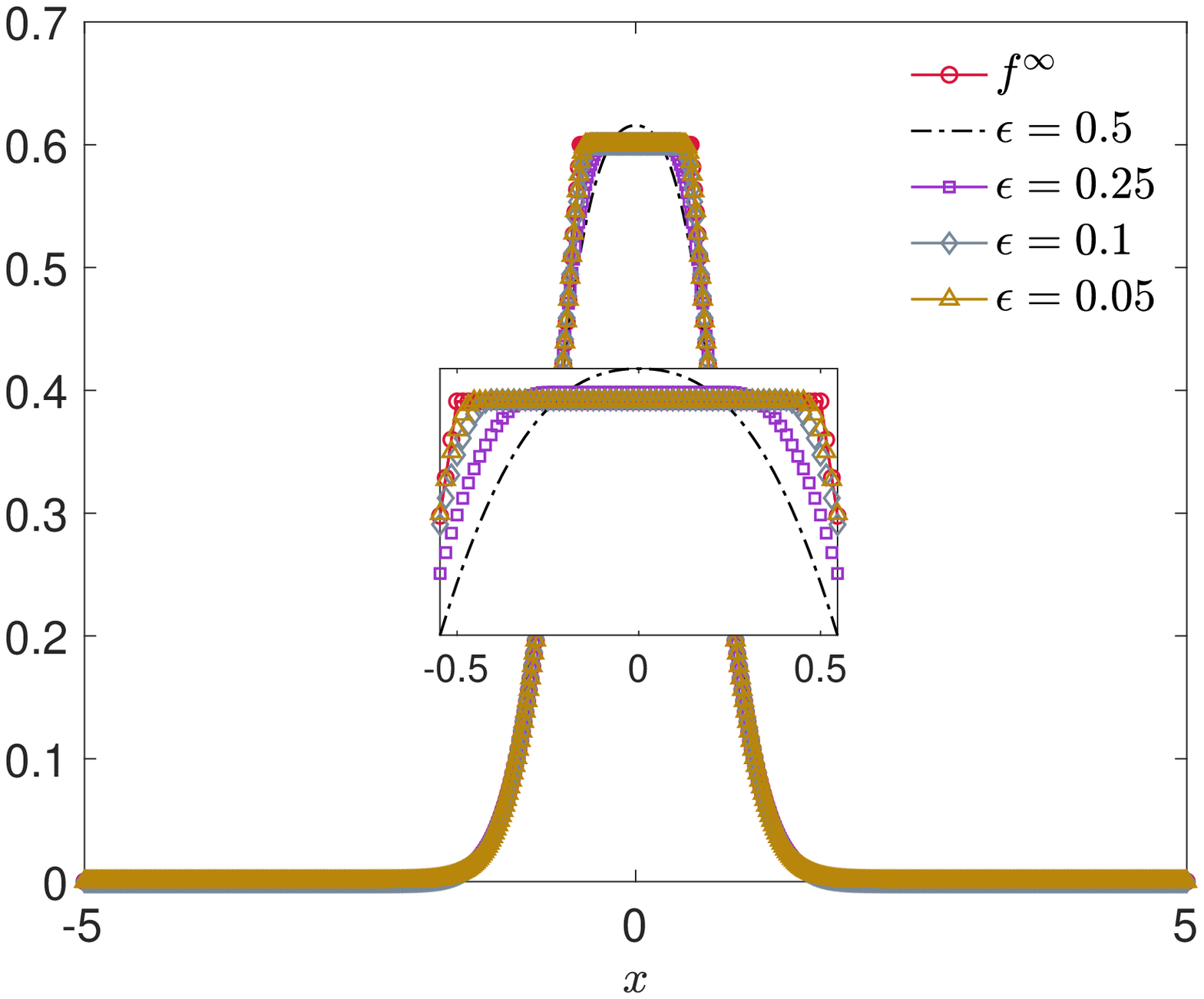}
\caption{Top and bottom left: evolution of $f_\epsilon(x,t)$ for $\epsilon = 0.5$ (top left), $\epsilon = 0.25$ (top right) and $\epsilon = 0.05$ (bottom left) and obtained through the SP scheme for the problem \eqref{eq:model_epsilon} over the time interval $[0,5]$ \rev{and initial distribution \eqref{eq:f0_num}}. We considered the domain $[-5,5]$ discretized by $N = 81$ gridpoints and $\Delta t = \Delta x^2/L^2$. Bottom right: comparison of the large time numerical solution of \eqref{eq:model_epsilon} at time $T= 10$, i.e. $f_\epsilon(x,T)$, with the analytical steady state $f^\infty(x)$ of the initial problem.   }
\label{fig:SP}
\end{figure}

\begin{figure}
\centering
\includegraphics[scale = 0.35]{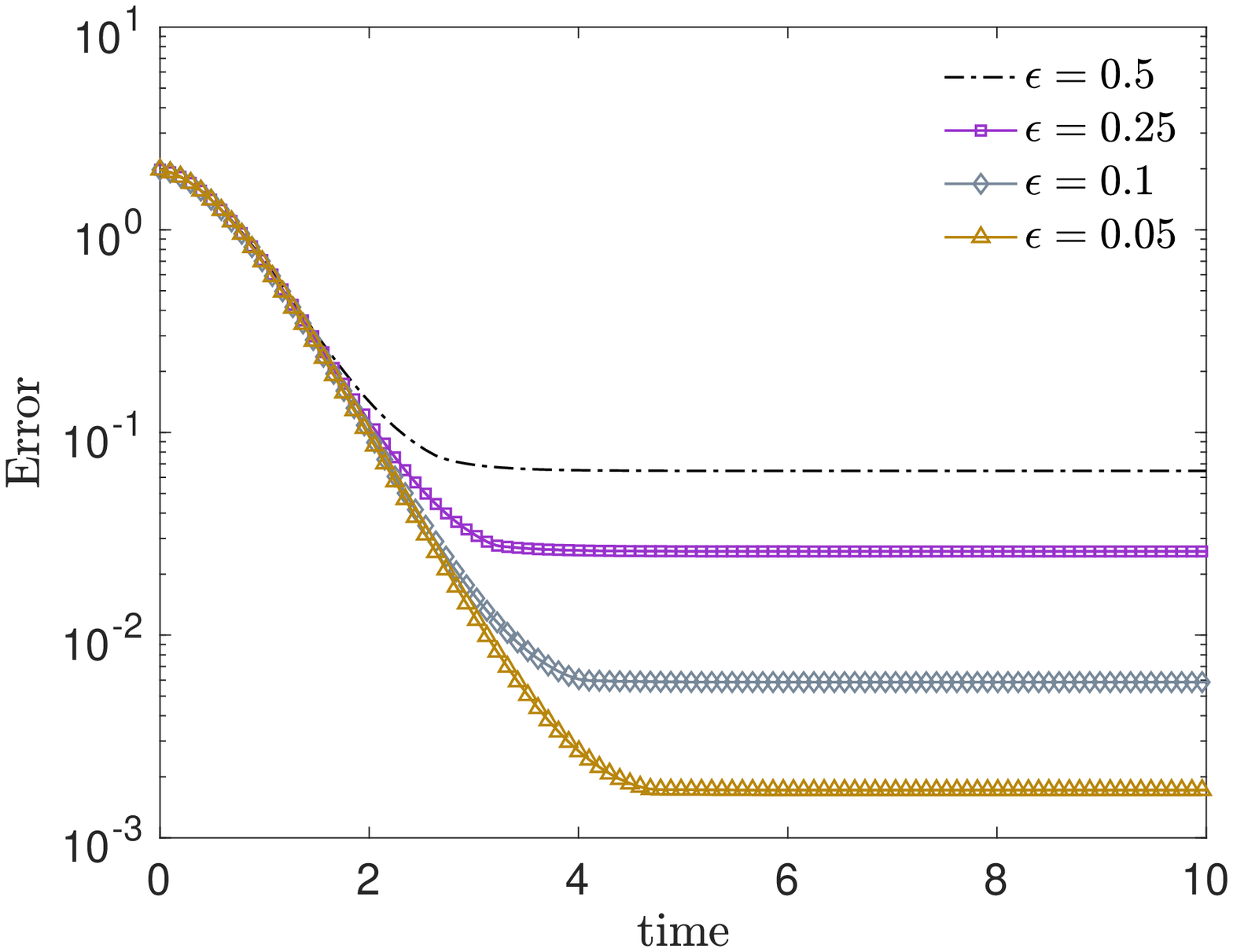}
\includegraphics[scale = 0.35]{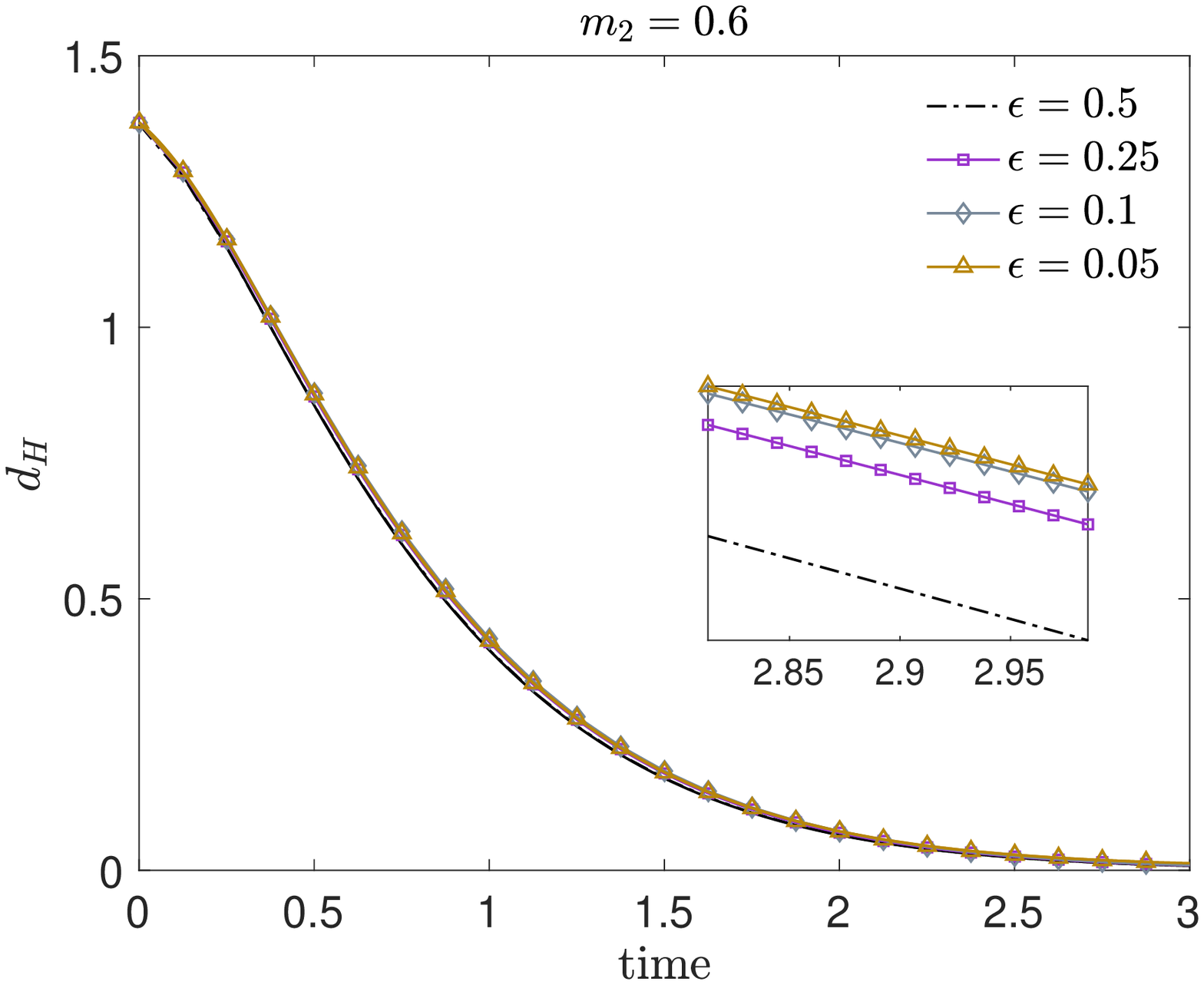}
\caption{\rev{Right:} evolution of the Hellinger distance $d_H$ over the time horizon $[0,3]$ for several choices of $\epsilon>0$ \rev{and computed for the problem \eqref{eq:model_epsilon} with $D = \{x \in \mathbb{R}: |x|\le \frac{1}{2}\}$ and $m_2 = 0.6$}. \rev{Left:} evolution of  $L^1$ relative error for several values of $\epsilon>0$.}
\label{fig:dH}
\end{figure}
In Figure \ref{fig:dH} we report the evolution of the Hellinger's distance $d_H$ defined in \eqref{eq:dH} for several values of $\epsilon= 0.05,0.1,0.25,0.5$. The distance $d_H$ has been computed with respect to a reference steady state \rev{corresponding to the case of a domain $D = \{x \in \mathbb R: |x |\le \frac{1}{2}\}$. The diffusion coefficient $\sigma^2>0$ has been determined as solution to \eqref{eq:sys2} with $m_2 = 0.6$}. In particular, we considered the numerical large time distributions $f_\epsilon(x,T)$, $T = 10$, for each $\epsilon>0$, computed with $N = 641$ gridpoints. In particular, for sufficiently small values of $\epsilon>0$, we get an approximation of $f^\infty(x)$ in \eqref{eq:finf} obtained through the surrogate model  with nonconstant diffusion \eqref{eq:model_epsilon} whose solution converges towards an $\epsilon$-dependent equilibrium. 

\rev{
\begin{figure}
\centering
\includegraphics[scale = 0.35]{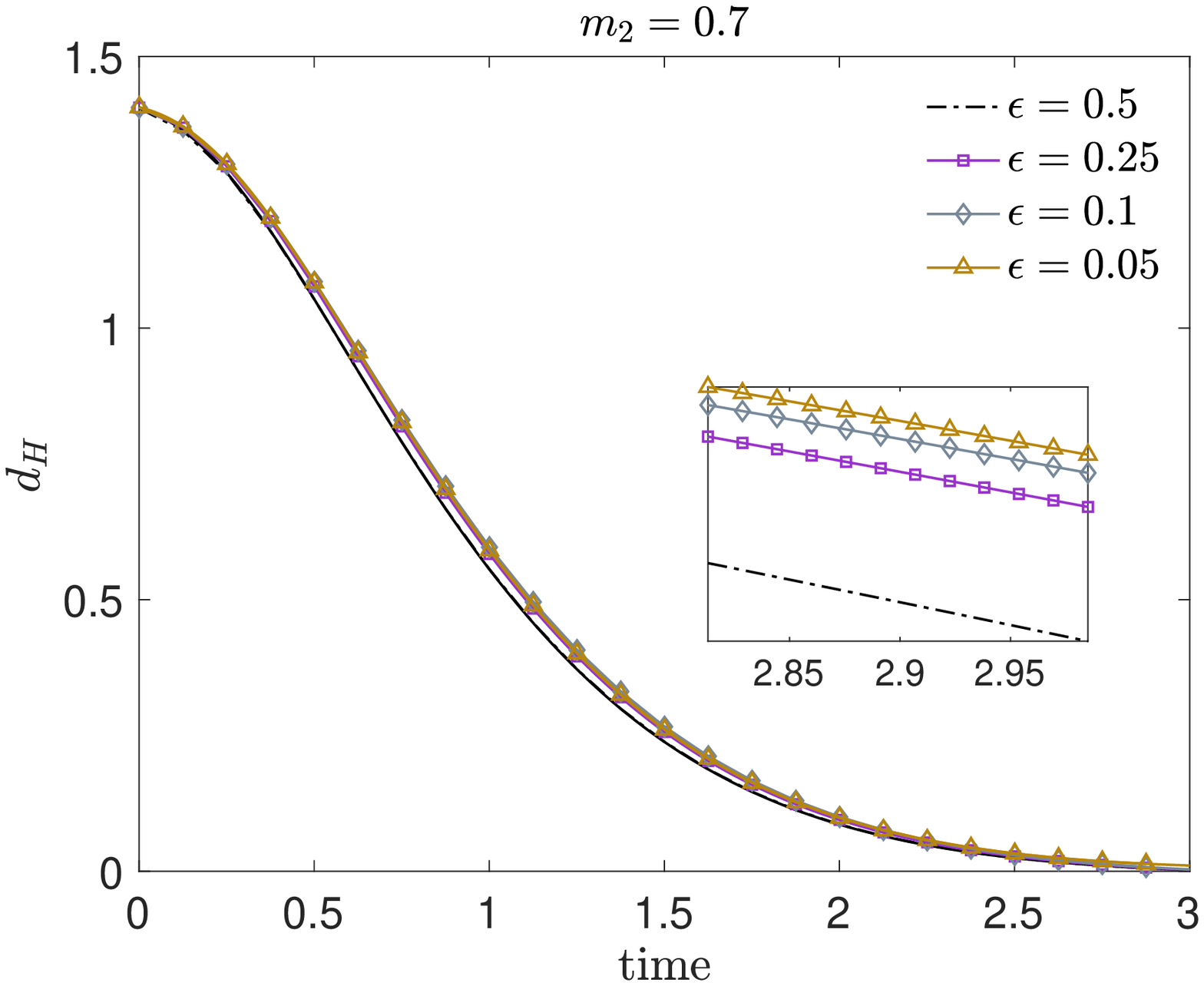}
\includegraphics[scale = 0.35]{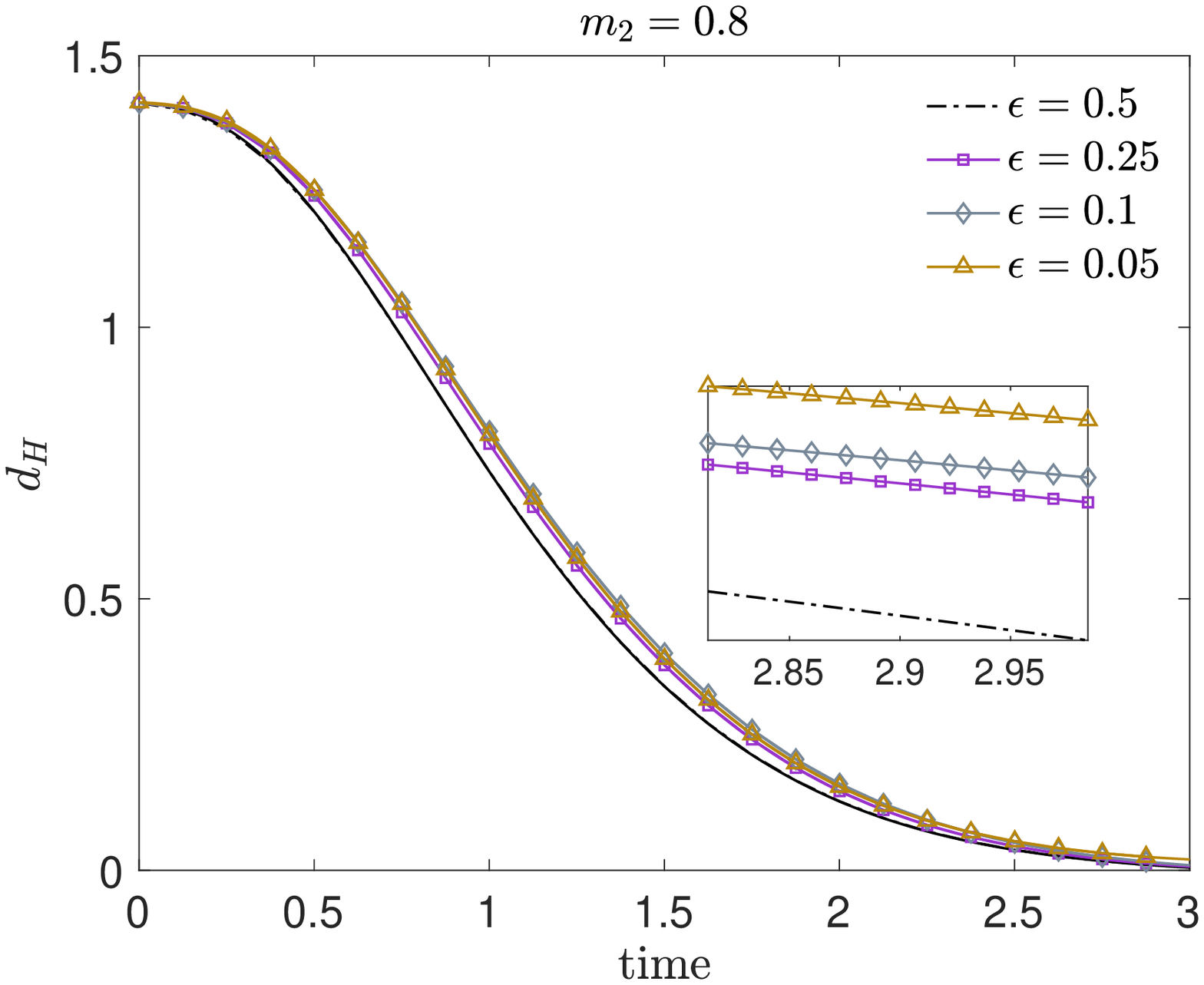}
\caption{\rev{Evolution of the Hellinger distance $d_H$ over the time horizon $[0,3]$ for several $\epsilon>0$ for the problem \eqref{eq:model_epsilon} and $m_2 = 0.7$ (left) or $m_2 = 0.8$ (right). }}
\label{fig:dH_m2}
\end{figure}
In Figure \ref{fig:dH_m2} we represent the evolution of the quantity $d_H$ computed as above with respect to a reference large time solution obtained at time $T = 10$ with $N = 641$ gridpoints and for two choices of the constant $m_2= 0.7,0.8$. As before, the domain $D$ corresponds to the set $\{x \in \mathbb R: |x| \le \frac{1}{2}\}$. We recall that $m_2>0$ is an input data of system \eqref{eq:sys2} and determines the fraction of mass in $D$, we value of $\sigma^2>0$ is determined as the (unique) solution of this system. We can observe that for increasing values of $m_2$, and therefore for vanishing values of $\sigma^2>0$ (see Figure \ref{fig:mass}), the rate of convergence decreases. }

\subsection{Extension to 2D}
In this section we explore numerically the convergence towards equilibrium of the 2D surrogate Fokker-Planck problem
\be
\label{eq:2D_epsi}
\partial_tf_\e(\x,t) = \nabla_\x \cdot \left[ (\x-\x_0)f_\e(\x,t) + \nabla_\x (\kappa_\epsilon(\x)f_\e(\x,t) )\right],
\ee
where  $\kappa_\epsilon(\x)$ is a suitable regularization of the radially symmetric discontinuous diffusion function \fer{eq:k} in 2D. 

In particular we will consider the radially symmetric regularization defined in \eqref{eq:kepsi}. For the 2D case, an extension of the SP scheme can be defined  based on a dimensional splitting approach. In the description of the scheme we avoid $\epsilon$ subscript of the density.  In particular, we introduce a uniform mesh $(x_i,y_j) \in [-L,L]\times [-L,L] \subset \mathbb R^2$ with $\Delta x = x_{i+1}-x_i>0$, $\Delta y = y_{i+1}-y_i>0$. Denoting with $f_{i,j}(t)$ the approximation of $f(x_i,y_j,t)$ we consider the discretization
\[
\dfrac{d}{dt} f_{i,j}(t) = \dfrac{\mathcal F_{i+1/2,j}[f]- \mathcal F_{i-1/2,j}[f]}{\Delta x} + \dfrac{\mathcal F_{i,j+1/2}[f]- \mathcal F_{i,j-1/2}[f]}{\Delta y}
\]
 where
 \[
 \begin{split}
 \mathcal F_{i+1/2,j}[f] &= \tilde{\mathcal C}_{i+1/2,j}\tilde{f}_{i+1/2,j} + \kappa_{\epsilon}(x_{i+1/2},y_j )\dfrac{f_{i+1,j}-f_{i,j}}{\Delta x} \\
  \mathcal F_{i,j+1/2}[f] &= \tilde{\mathcal C}_{i,j+1/2}\tilde{f}_{i,j+1/2} + \kappa_{\epsilon}(x_i,y_{j+1/2}) \dfrac{f_{i,j+1}-f_{i,j}}{\Delta y}
 \end{split}
 \]
and we set 
  \[
 \begin{split}
 \tilde{\mathcal C}_{i+1/2,j} &= \dfrac{\kappa_{\epsilon}(x_{i+1/2},y_j)}{\Delta x}\int_{x_i}^{x_{i+1}}\dfrac{(\mathbf x-\mathbf x_0) + \partial_x \kappa_\epsilon(x,y_j)}{\kappa_\epsilon(x,y)}dx \\
  \tilde{\mathcal C}_{i,j+1/2} &= \dfrac{\kappa_{\epsilon}(x_i,y_{j+1/2})}{\Delta y}\int_{y_j}^{y_{j+1}}\dfrac{(\mathbf x-\mathbf x_0) + \partial_y \kappa_\epsilon(x_i,y)}{\kappa_\epsilon(x,y)}dy
   \end{split}
 \]
  \[
 \begin{split}
 \tilde{f}_{i+1/2,j} &= (1-\delta_{i+1/2,j})f_{i+1,j} + \delta_{i+1/2,j}f_{i,j}\\
  \tilde{f}_{i,j+1/2} &= (1-\delta_{i,j+1/2})f_{i,j+1} + \delta_{i,j+1/2}f_{i,j}
  \end{split}
 \]
 with nonlinear weights defined as follows
 \[
 \begin{split}
 \delta_{i+1/2,j} &= \dfrac{1}{\lambda_{i+1/2,j}} + \dfrac{1}{1-\exp(\lambda_{i+1/2,j})} \\
  \delta_{i,j+1/2} &= \dfrac{1}{\lambda_{i,j+1/2}} + \dfrac{1}{1-\exp(\lambda_{i,j+1/2})}
 \end{split}
 \]
 being $\lambda_{i+1/2,j} = \dfrac{\Delta x \, \tilde{\mathcal C}_{i+1/2,j}}{\kappa_{\epsilon}(x_{i+1/2},y_j)}$ and  $\lambda_{i,j+1/2} = \dfrac{\Delta y \, \tilde{\mathcal C}_{i,j+1/2}}{\kappa_{\epsilon}(x_i,y_{j+1/2})}$.

We consider as initial distribution 
\begin{equation}
\label{eq:f0_2D}
\begin{split}
f(\x,0) =   \dfrac{1}{8\pi \theta^2}\exp\left\{ -\dfrac{(x - \mu_x)^2}{2\theta^2} -\dfrac{(y - \mu_y)^2}{2\theta^2} \right\} + \dfrac{1}{8\pi \theta^2}\exp\left\{ -\dfrac{(x - \mu_x)^2}{2\theta^2} -\dfrac{(y + \mu_y)^2}{2\theta^2}\right\} \\
+\dfrac{1}{8\pi \theta^2}\exp\left\{ -\dfrac{(x + \mu_x)^2}{2\theta^2} -\dfrac{(y - \mu_y)^2}{2\theta^2} \right\} + \dfrac{1}{8\pi \theta^2}\exp\left\{ -\dfrac{(x + \mu_x)^2}{2\theta^2} -\dfrac{(y + \mu_y)^2}{2\theta^2}\right\}
\end{split}\end{equation}
with $\mu_x=3$, $\mu_y = -3$, $\theta^2 = 0.2$. We introduce a uniform discretization of the domain $[-L,L] \times [-L,L] $, $L = 5$, obtained with $N_x = N_y$ equally spaced gridpoints such that $\Delta x = \Delta y = 2L/(N-1)$. The time integration is performed through a standard RK4 method with $\Delta t = \Delta x^2/L^2$. The evolution of the density $f_\epsilon$ is presented in Figure \ref{fig:f_evo} where we considered $\epsilon = 0.05,0.5$ and we represent the 2D distribution at times $t = 1$ and $t = 10$. 

In Figure \ref{fig:surf} we represent the large time  solution of the surrogate model \eqref{eq:2D_epsi} for several $\epsilon = 0.05,0.1,0.25,0.5$ compared with the equilibrium distribution $f^\infty(\x)$ in \eqref{eq:eq_mD} with $\sigma^2 = 0.2555$, $\delta = 1$ and $\x_0 = (0,0)$. Furthermore, in Figure \ref{fig:2D2} (left plot) we represent the obtained marginal distributions at time $T = 10$, i.e. $\int_{\mathbb R}f_\epsilon(\x,T)dx$, compared with the analytical marginal distribution $\int_{\mathbb R}f^\infty(\x)dx$. We may easily observe how for small $\epsilon>0$ the large time solution of the surrogate problem is capable to approximate the analytical one. 

\begin{figure}
\centering
\subfigure[$t = 1,\epsilon = 0.5$]{\includegraphics[scale = 0.36]{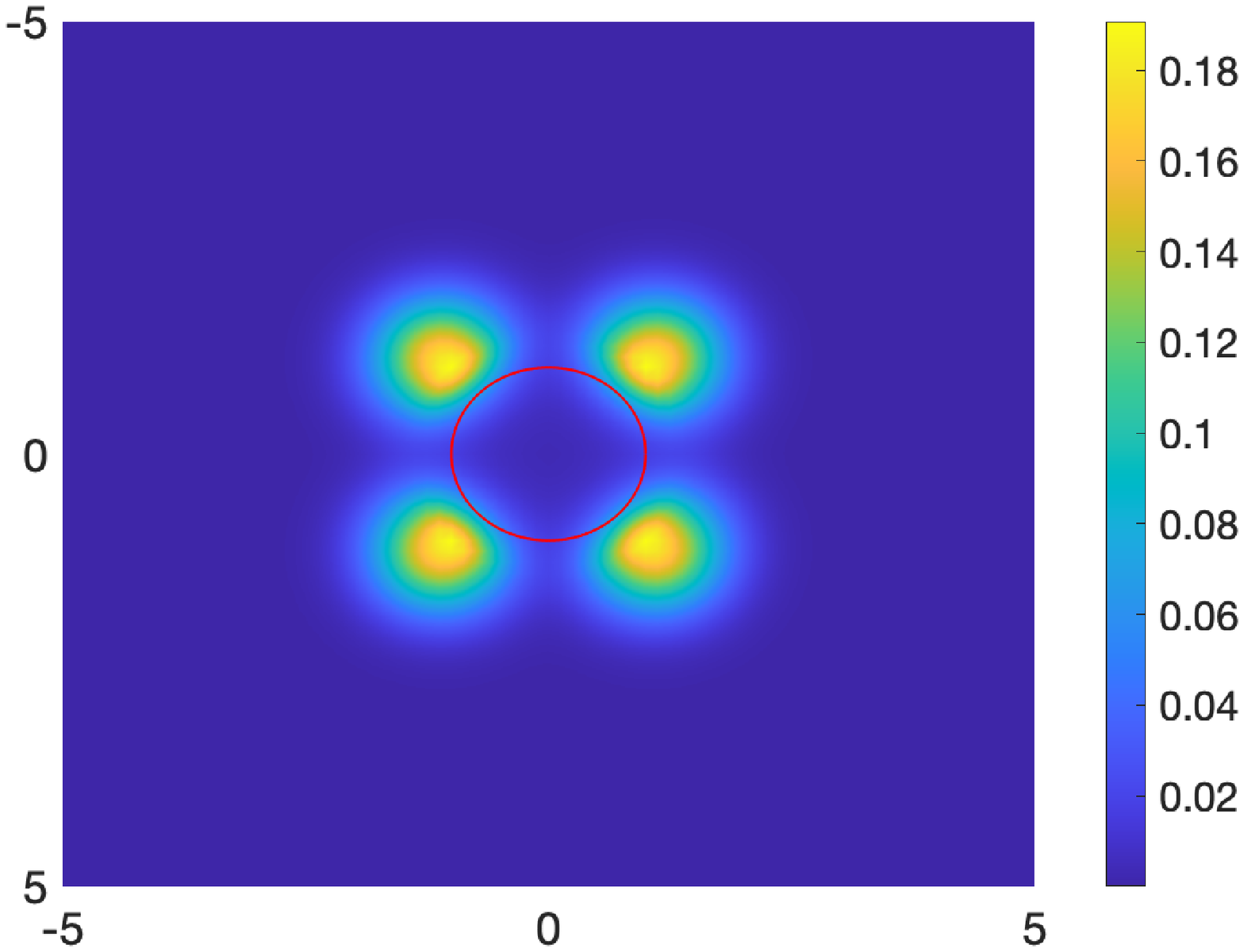}}
\subfigure[$t = 10,\epsilon = 0.5$]{\includegraphics[scale = 0.36]{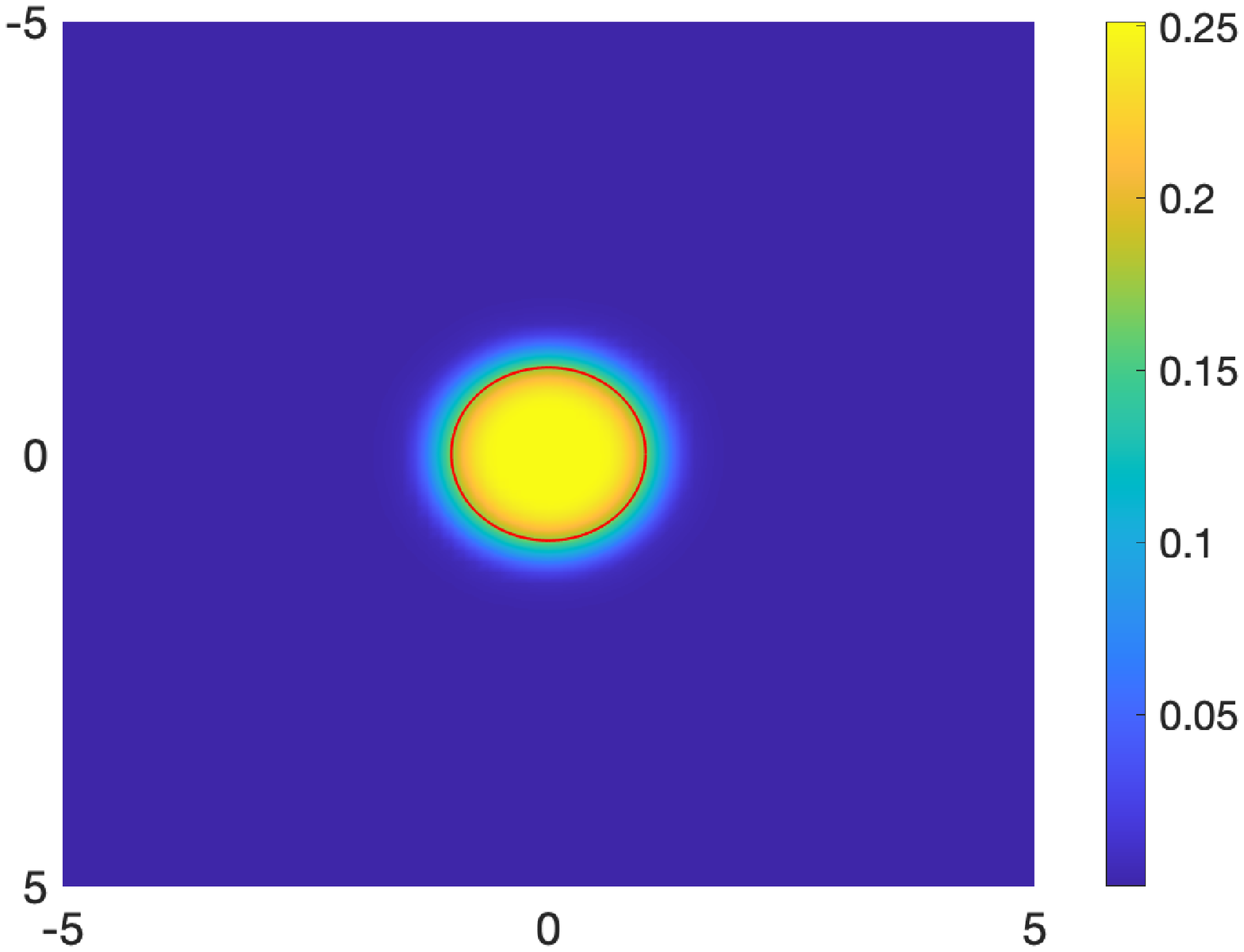}}\\
\subfigure[$t = 1,\epsilon = 0.05$]{\includegraphics[scale = 0.36]{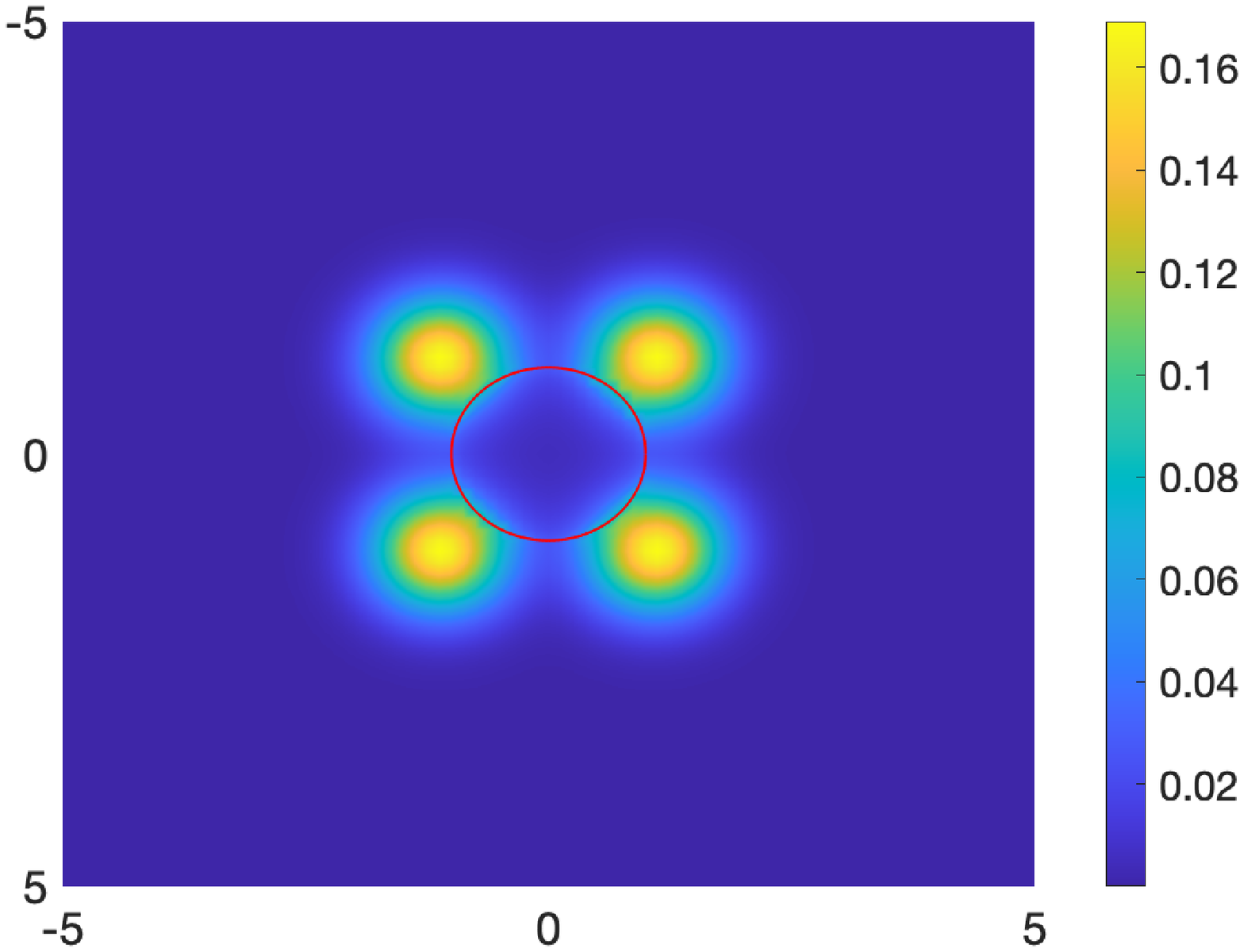}}
\subfigure[$t = 10,\epsilon = 0.05$]{\includegraphics[scale = 0.36]{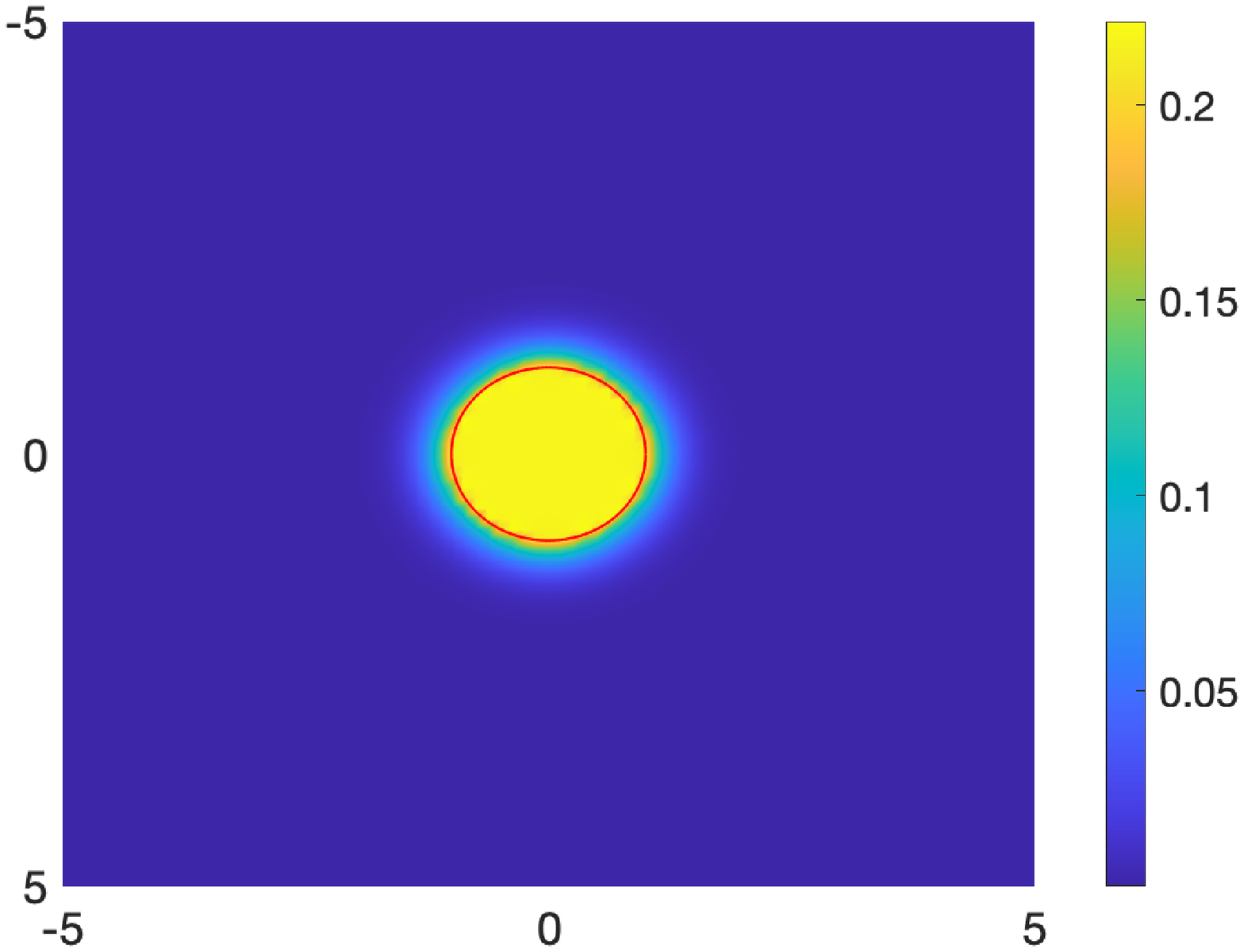}}
\caption{Evolution of the numerical large time solution of problem \eqref{eq:2D_epsi}  for several $\epsilon= 0.5$ (top row) and $\epsilon = 0.05$ (bottom row). We considered $[-L,L]^2$, $L=5$, and a discretization of $N_x = N_y = 81$ gridpoints in each dimension, $\Delta t = \Delta x/L^2$. We visualize the target domain $D = \{\mathbf{x} \in \mathbb{R}^2: |\mathbf{x} - \mathbf{x}_0|\le \delta\}$, $\delta = 1$, $\mathbf{x}_0= 0$ in red.  }
\label{fig:f_evo}
\end{figure}

\begin{figure}
\centering
\subfigure[$\epsilon = 0.5$]{\includegraphics[scale = 0.36]{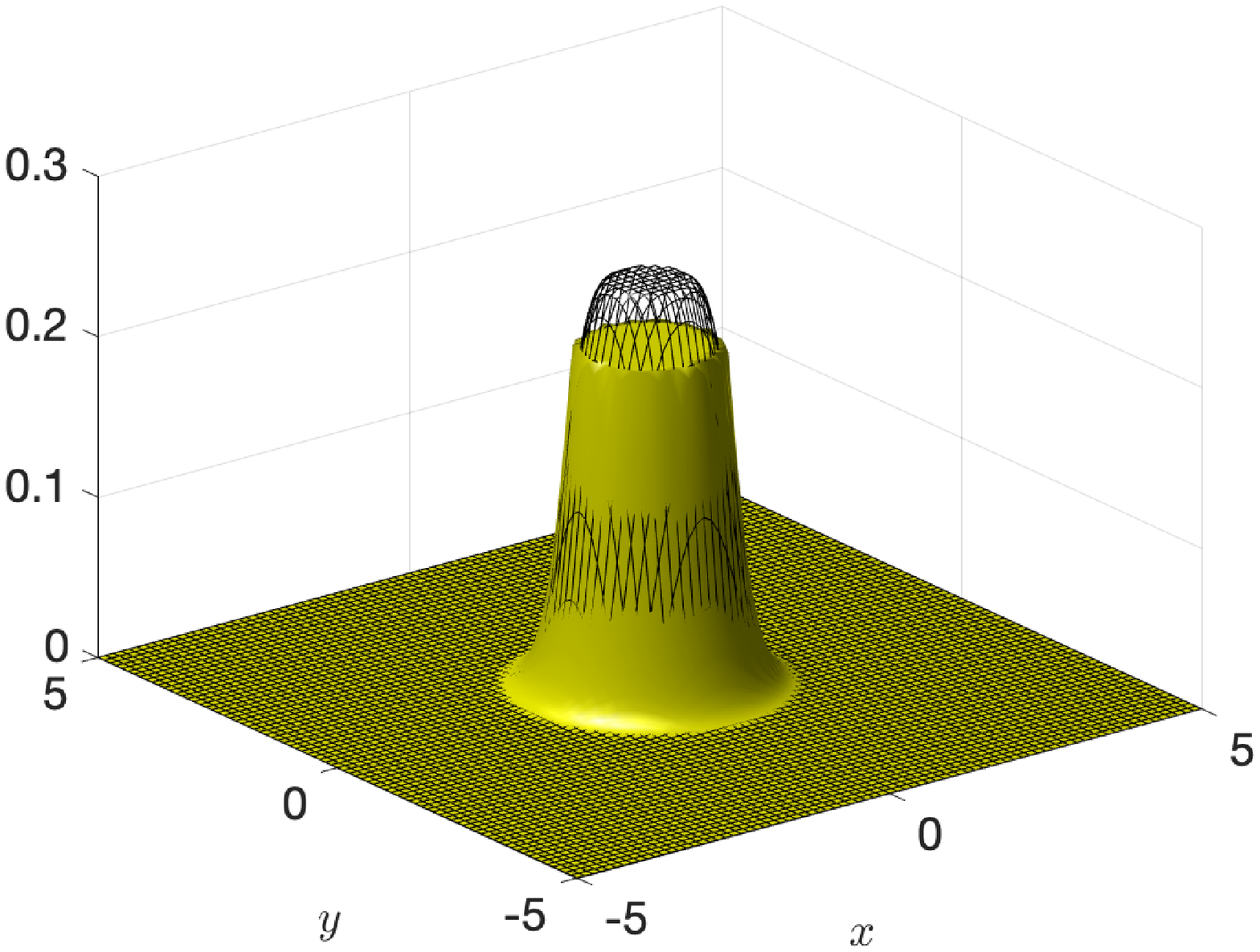}}
\subfigure[$\epsilon = 0.25$]{\includegraphics[scale = 0.36]{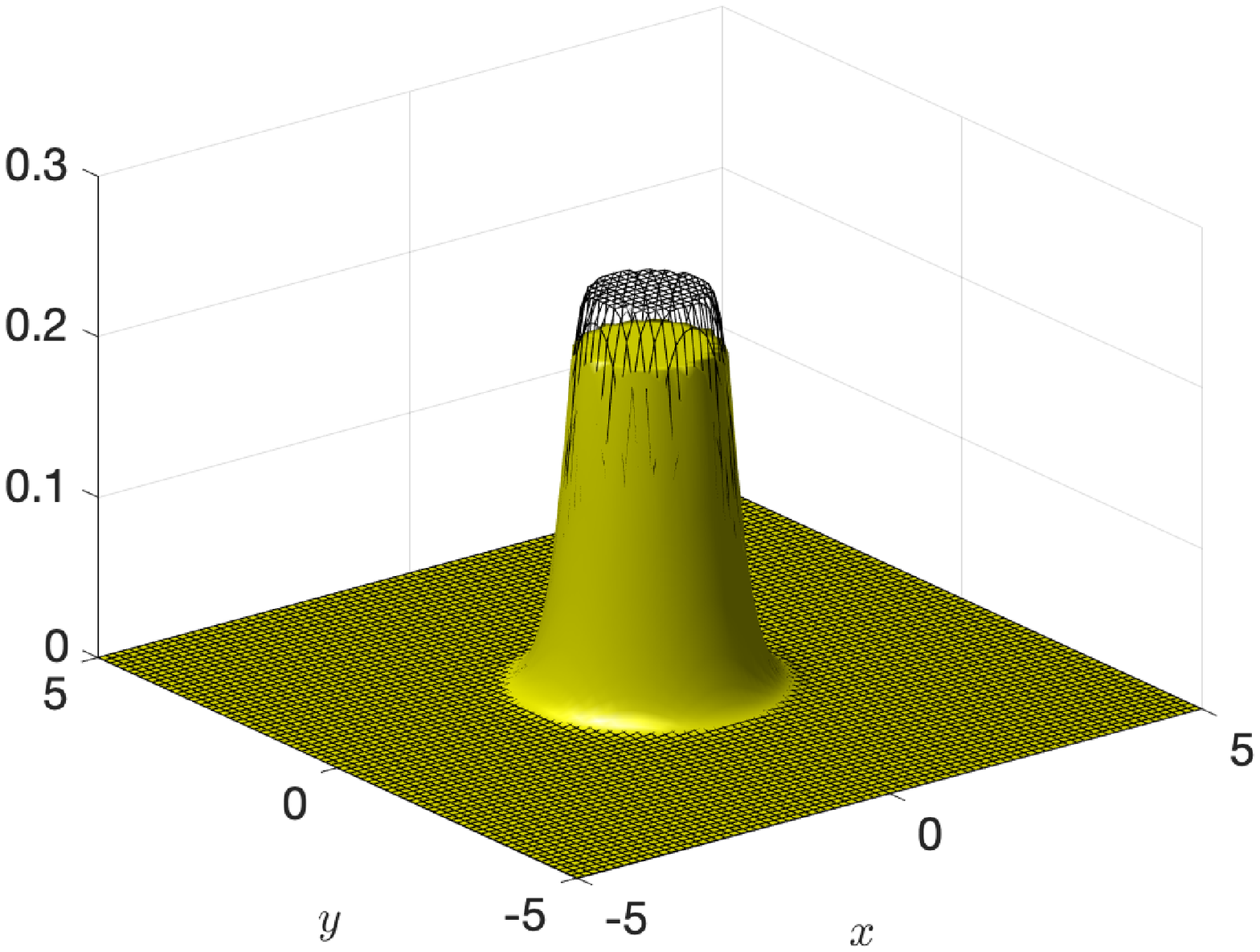}}\\
\subfigure[$\epsilon = 0.1$]{\includegraphics[scale = 0.36]{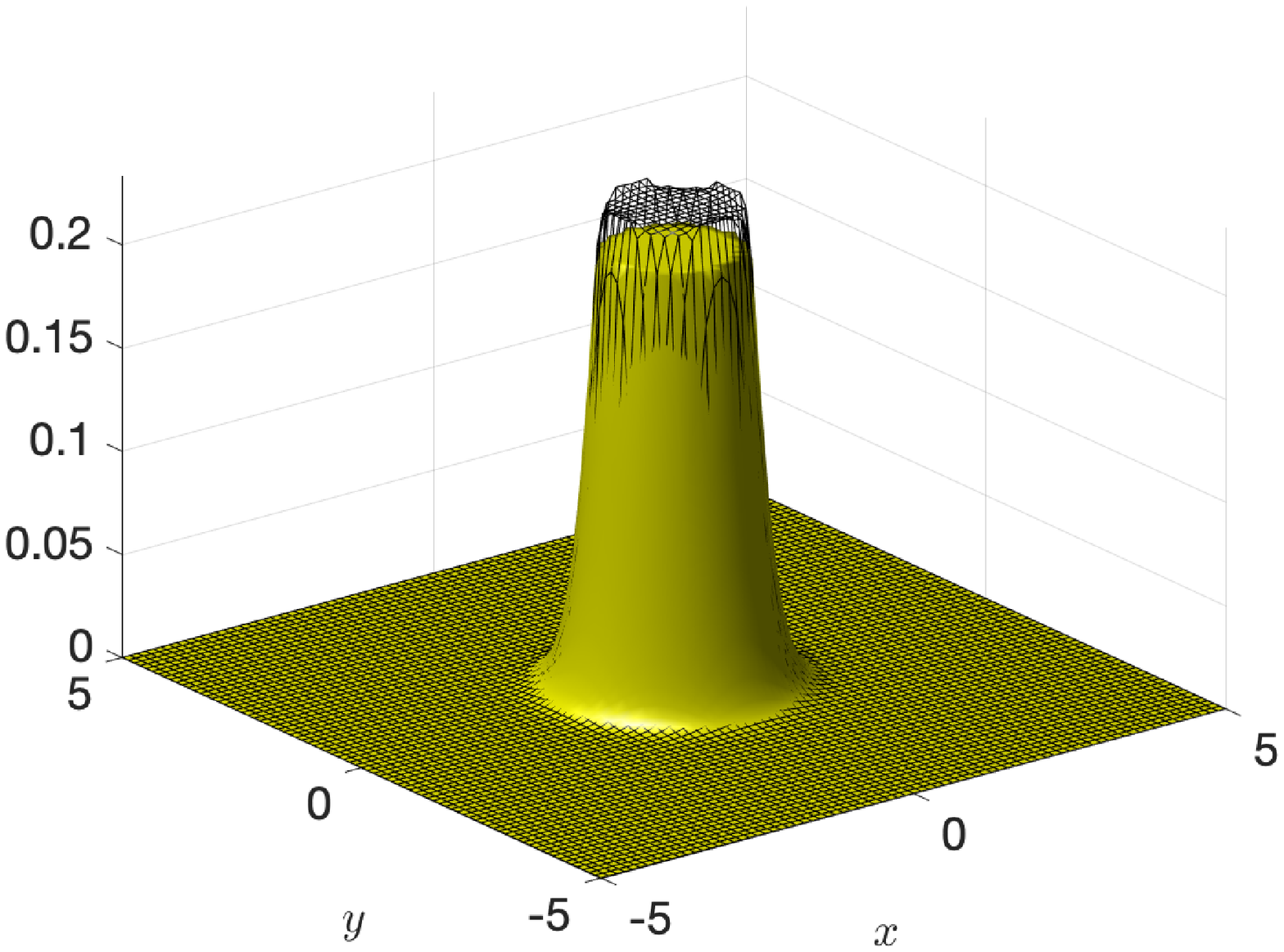}}
\subfigure[$\epsilon = 0.05$]{\includegraphics[scale = 0.36]{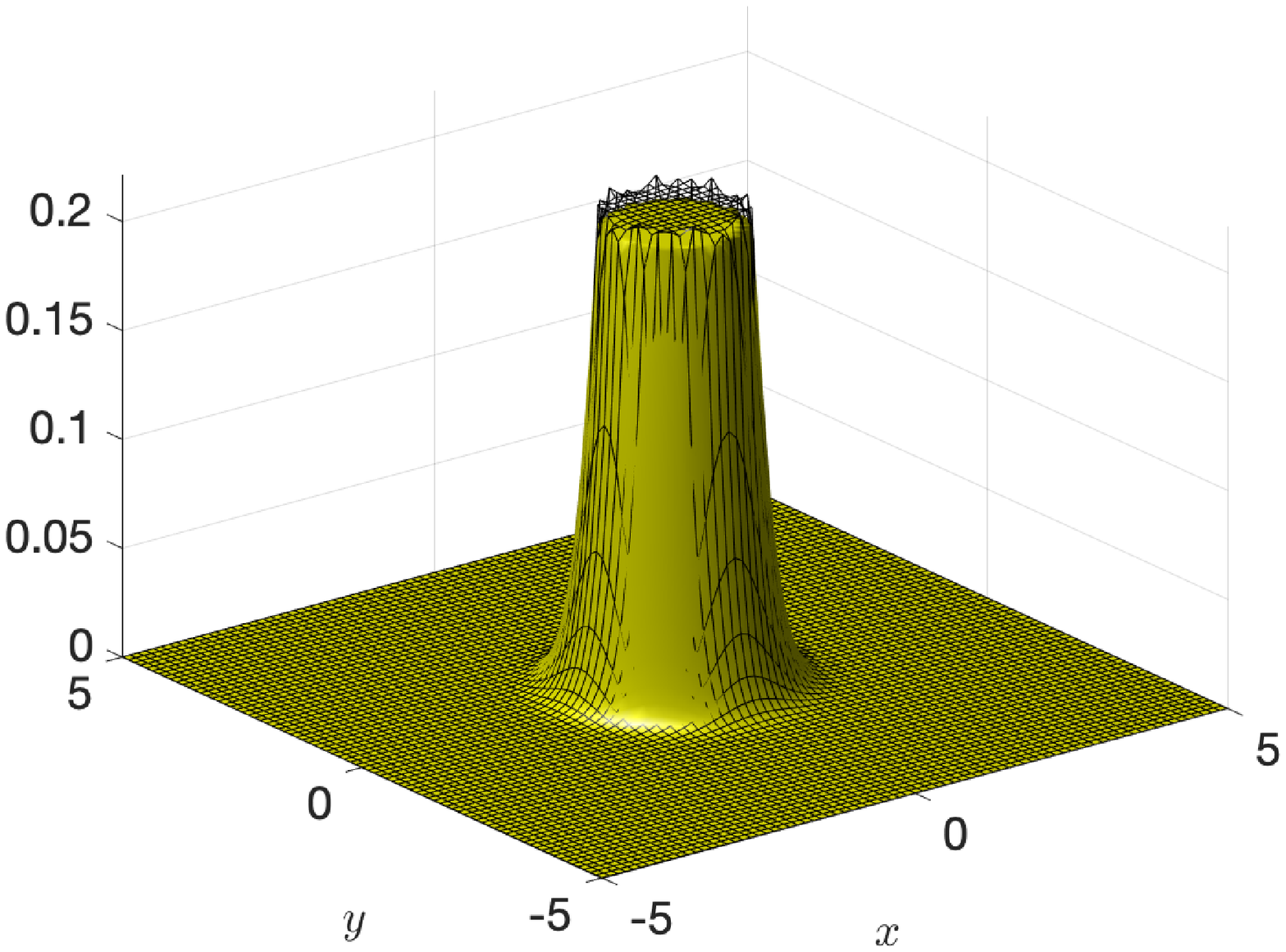}}
\caption{Comparison of the numerical large time solution of problem \eqref{eq:2D_epsi}  for several $\epsilon >0$ with the analytical one defined in \eqref{eq:equi2}. We considered $[-L,L]^2$, $L=5$, and a discretization of $N_x = N_y = 81$ gridpoints in each dimension, $\Delta t = \Delta x/L^2$. We visualize the numerical solution through the black mesh at time $T = 10$ and the analytical equilibrium with the yellow surface.  }
\label{fig:surf}
\end{figure}

Finally, in Figure \ref{fig:2D2} (right plot) we compute the evolution of the 2D Hellinger's distance 
\be
\label{eq:dH2D}
d_H^2(f,g)(t) = \int_{\mathbb R^2} (\sqrt{f(\x,t)} - \sqrt{g(\x)})^2d\x. 
\ee
for several $\epsilon>0$. The distance $d_H$ has been computed with respect to the reference large time solution of \eqref{eq:2D_epsi} with $N = 641$ gridpoints in both space dimensions at time $T = 50$. Hence, we considered $g = \bar f_\epsilon(\mathbf{x},T)$ computed on the introduced refined grid. 

\begin{figure}
\centering
\includegraphics[scale =0.36]{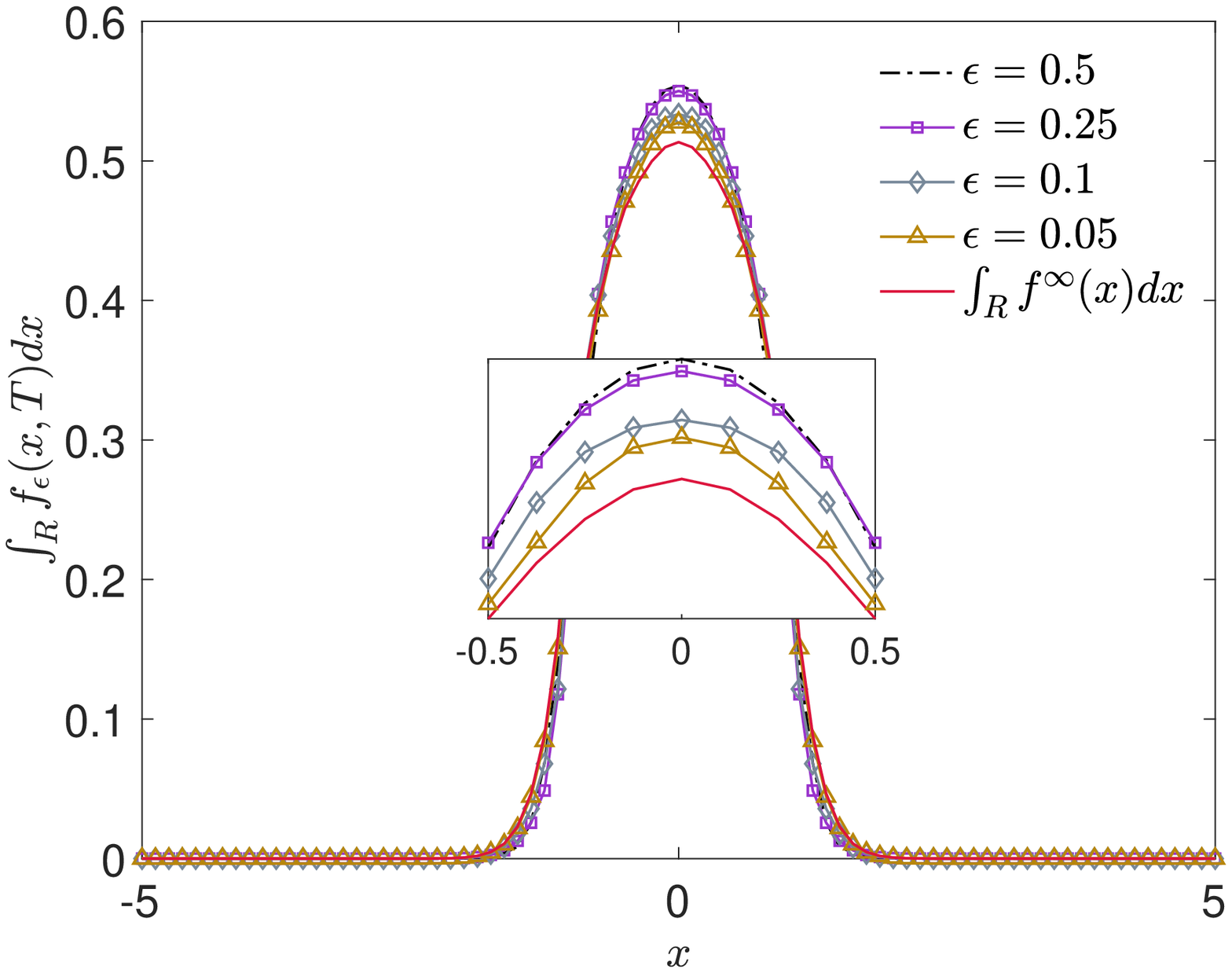}
\includegraphics[scale =0.36]{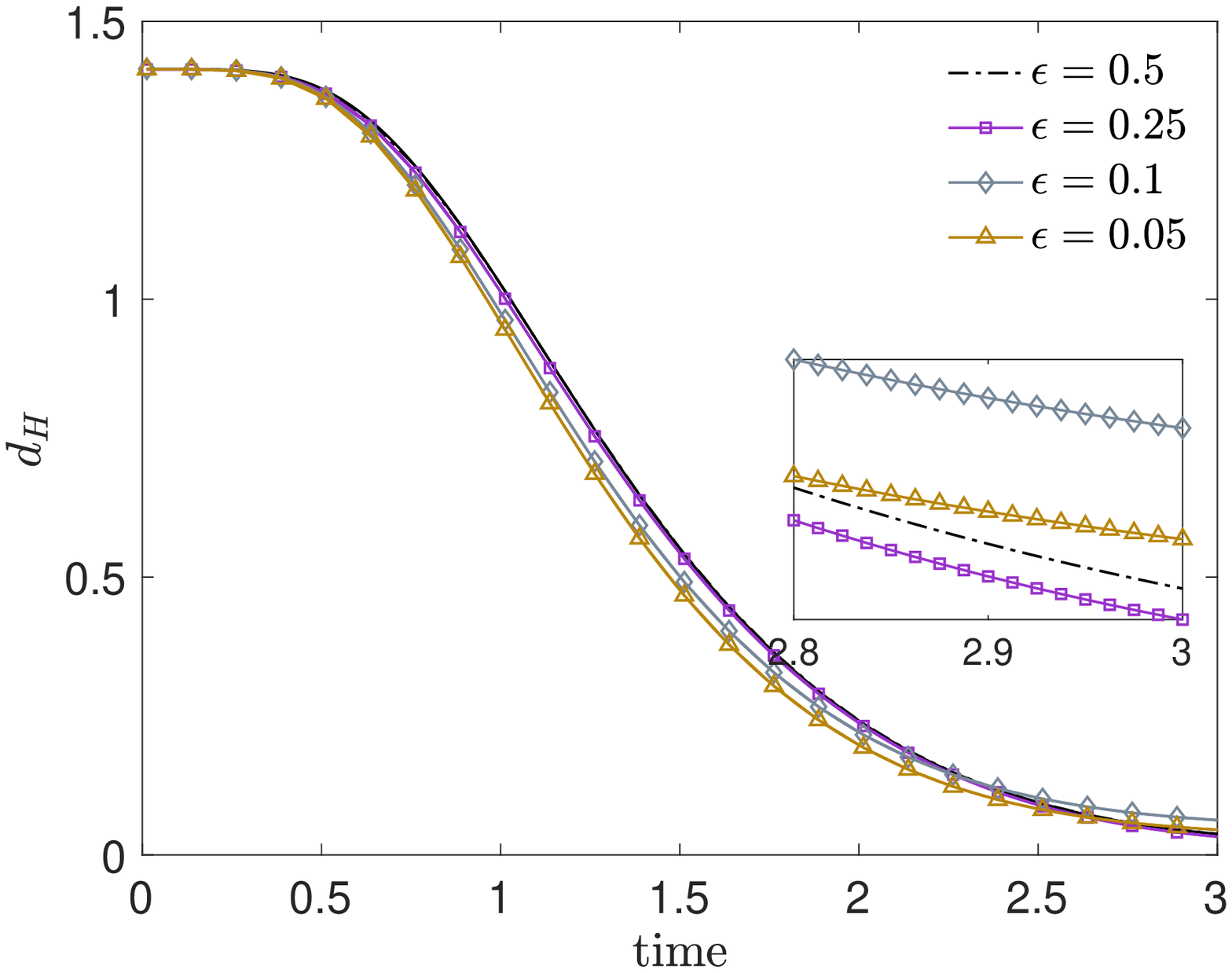}
\caption{Left: comparison of the marginal of $f^\infty(\x)$ with the marginals $\int_{\mathbb R}f_\epsilon(\x,T)dx$ obtained from the numerical solution of \eqref{eq:2D_epsi} at time $T = 10$. Right: evolution of the Hellinger's distance \eqref{eq:dH2D} computed through a reference large times solution produced with $N_x = N_y = 641$ at time $T = 10$. The positivity of the solution is guaranteed by keeping $\Delta t = \Delta x^2/L^2$. }
\label{fig:2D2}
\end{figure}

\subsection{\rev{Convergence of the particles' system}}
\rev{
In this section we investigate the convergence of a system of particles whose distribution corresponds to the introduced surrogate Fokker-Planck equation with the regularized diffusion weight $\kappa_\epsilon(x)$ defined in \eqref{eq:kepsi}.  In particular we consider the a system of SDEs describing the particles' position $\x_i(t)$ given by
\begin{equation}
\label{eq:SDE}
d\x_i(t) = (\x_0 - \x_i)dt + \sqrt{2\kappa_\epsilon(\x)}d\textbf{W}_i^t, 
\end{equation}
where $i = 1,\dots,N$  
In \eqref{eq:SDE} we denoted by $\{\textbf{W}_i^t\}_{i=1}^N$ a vector of $N$ independent $d$-dimensional Wiener processes. We recall that the transition to chaos for the particle system \eqref{eq:SDE} follows from the arguments in \cite{BCC,CFRT,CFTV}. We evolve over the time interval $[0,T]$, $T = 5$, the introduced system by means of the Euler-Maruyama method with $\Delta t = 10^{-2}$. We are interested in the approximation of the large time behaviour of the particles' distribution that is here reconstructed by means of a simple histogram in the domain $[-5,5]$ discretized with $N_x = 101$ gridpoints.  In the following we will fix a mass fraction $m_2 = 0.8$, meaning that, on average, $80\%$ of the particles must lie inside the domain $D$.  In Figure \ref{fig:part} we depict the reconstructed distribution of the particles' system \eqref{eq:SDE} of size $N = 10^5$ for two choices of $\epsilon = 0.05,0.5$. The initial positions of the particles is sampled from \eqref{eq:f0_num}. We may observe that for decreasing values of $\epsilon$ we approximate correctly the theoretical distribution $f^\infty(x)$ defined in  \eqref{eq:finf} of the domain $D = \{x \in \mathbb R: |x|<\frac{1}{2}\}$ and relative to the choice $m_2 = 0.8$. The parameters characterizing diffusion and continuity at the interface, i.e. $\sigma^2>0$ and $m_1>0$, have been determined as solution to \eqref{eq:sys2}. In the right plot of Figure \ref{fig:part} we present the evolution of the number of particles in $D$ for several values of $\epsilon>0$. As expected, for large times and for vanishing values of $\epsilon>0$, the number of particles approximate the theoretical level $m_2 = 0.8$, meaning that the effective number of particles in $D$ coincides with the theoretical one. 
\begin{figure}
\centering
\includegraphics[scale = 0.36]{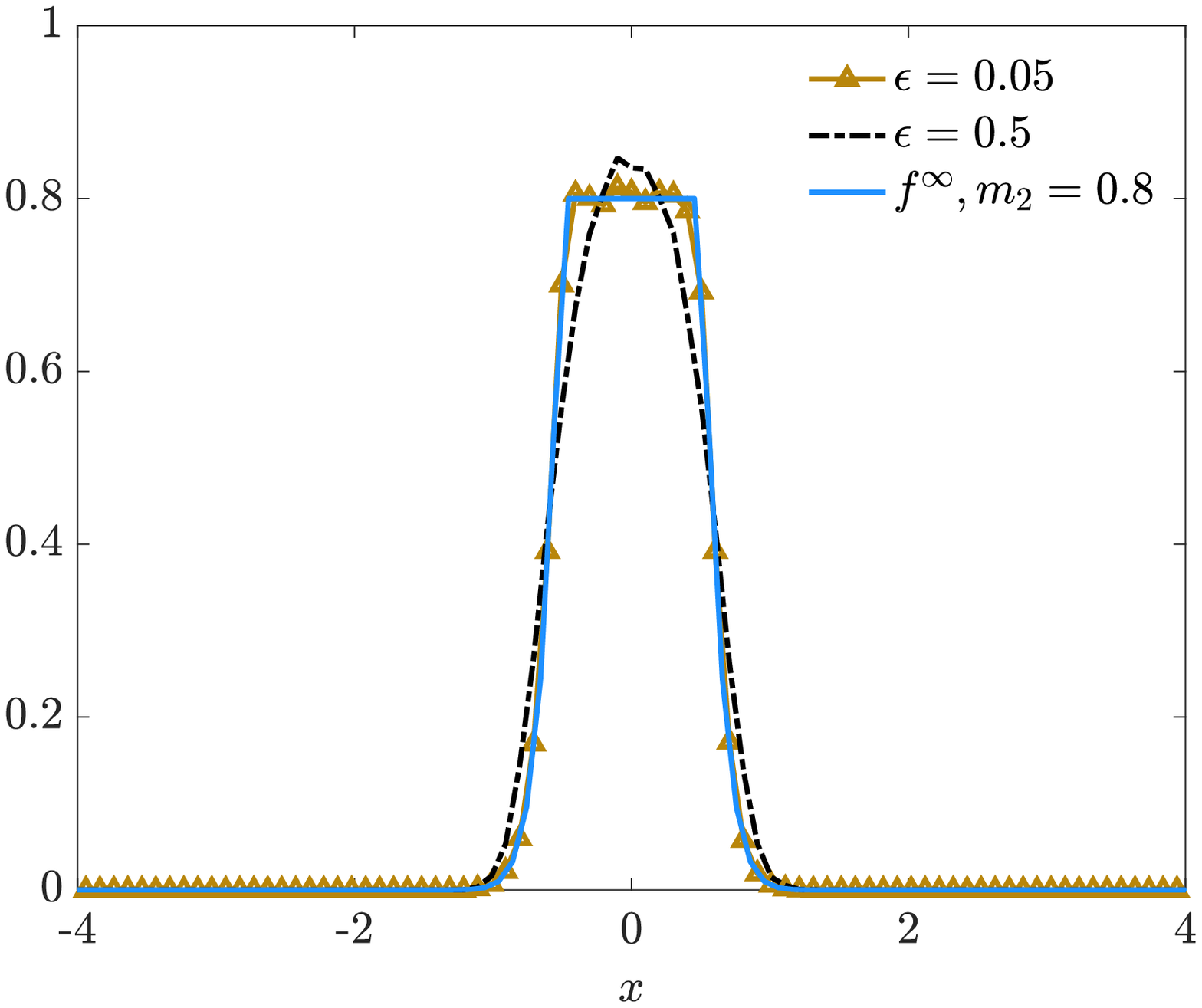}
\includegraphics[scale = 0.36]{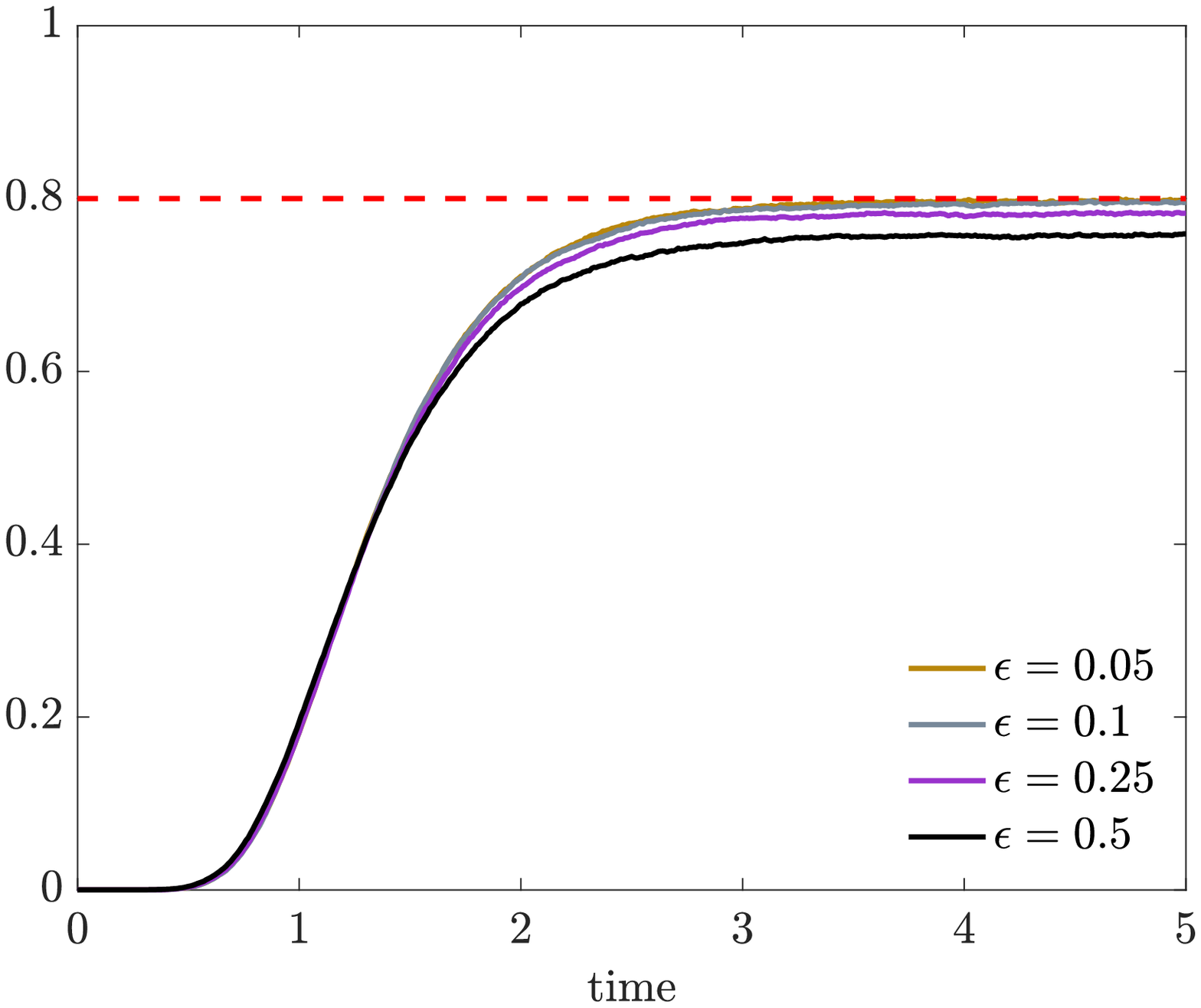}
\caption{\rev{Left: comparison between the theoretical $f^\infty(x)$ for $m_2 = 0.8$ and the density reconstruction of the particles' system \eqref{eq:SDE} in 1D at time $T = 5$, $N = 10^5$, where $D = \{x \in \mathbb R: |x|<\frac{1}{2}\}$. We considered two values of $\epsilon = 0.05,0.5$ and the regularization of the diffusion function $\kappa_\epsilon(x)$ in \eqref{eq:kepsi}. Right: evolution of the mass fraction in $D$ for several $\epsilon>0$. }}
\label{fig:part}
\end{figure}
Analogous results are obtained in the 2D case. In this direction, in Figure \ref{fig:part_2D} we consider a 2D particles' system of the form \eqref{eq:SDE} composed by $N = 10^3$ particles. We assume $D = \{\x \in \mathbb R^2: |\x|\le \frac{1}{2}\}$ and  $m_2 = 0.8$. The parameters $\sigma^2,m_1>0$ have been determined as solution to \eqref{eq:prob_2D}. The evolution of the system is presented for $t = 0,1,2,4$ and has been obtained with an Euler-Maruyama method with $\Delta t = 10^{-2}$. The initial positions have been sampled from \eqref{eq:f0_2D} with $\mu_x = -\mu_y = 3$, $\theta^2 = 0.2$. 
\begin{figure}
\centering
\includegraphics[scale = 0.36]{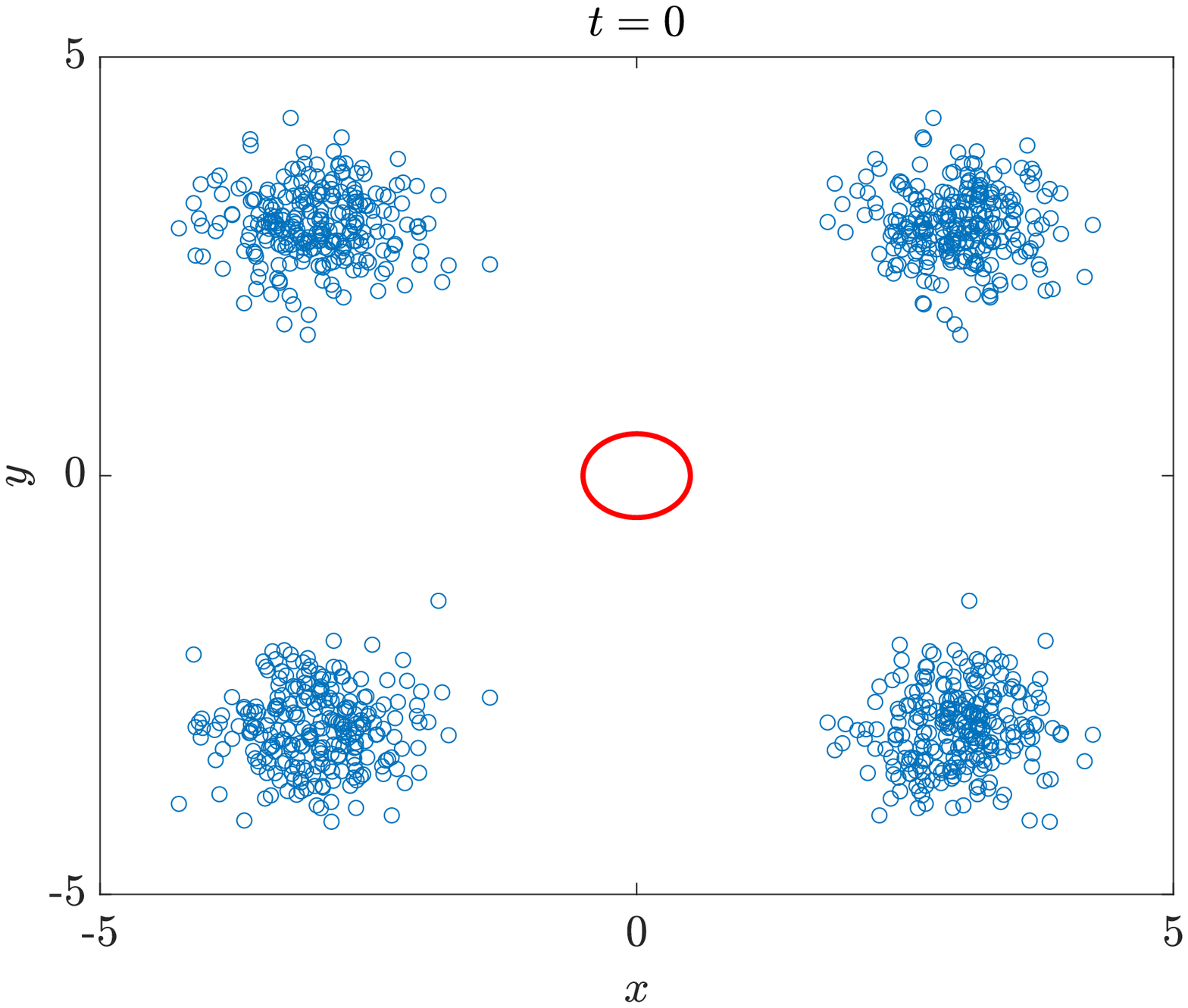}
\includegraphics[scale = 0.36]{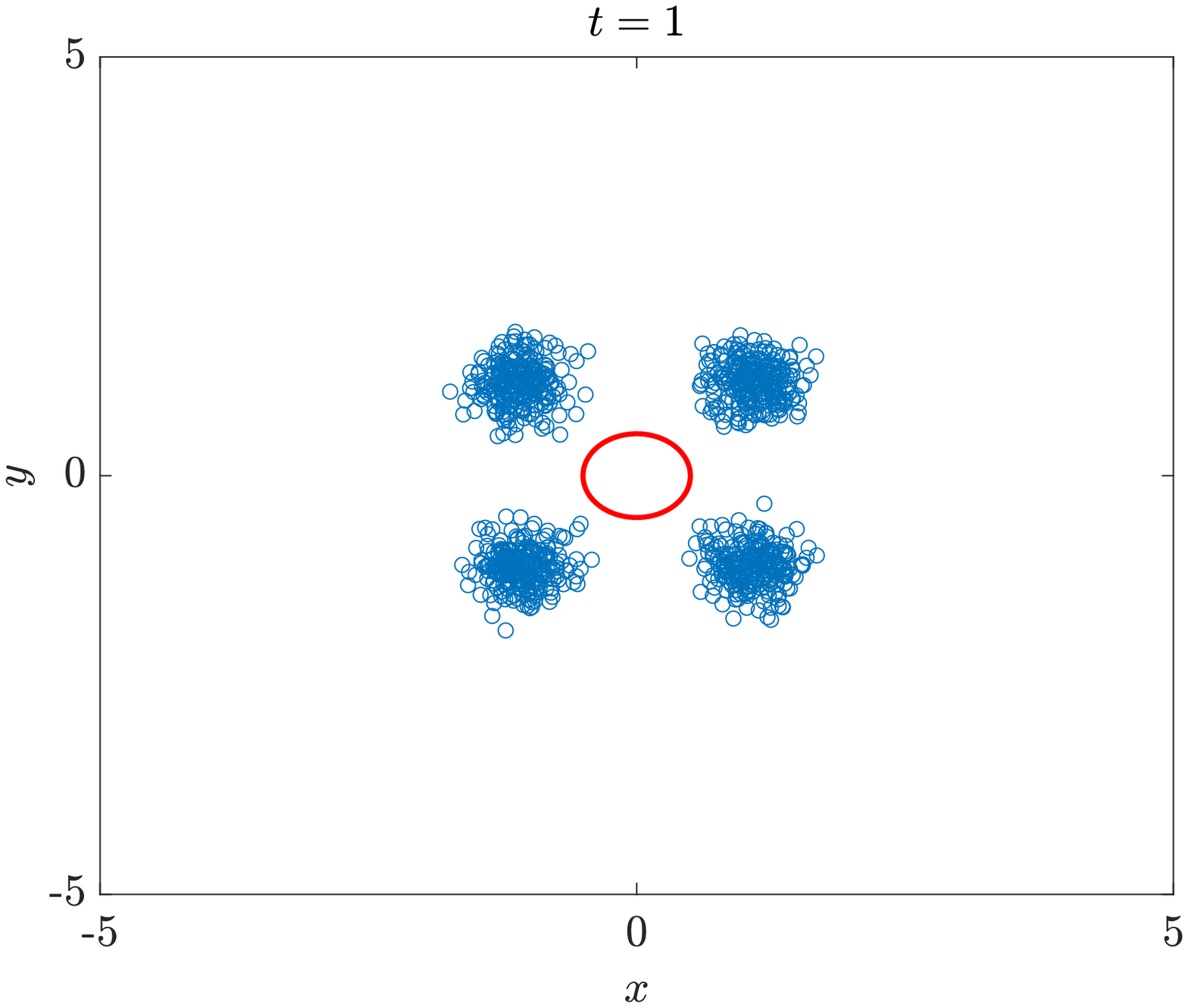}\\
\includegraphics[scale = 0.36]{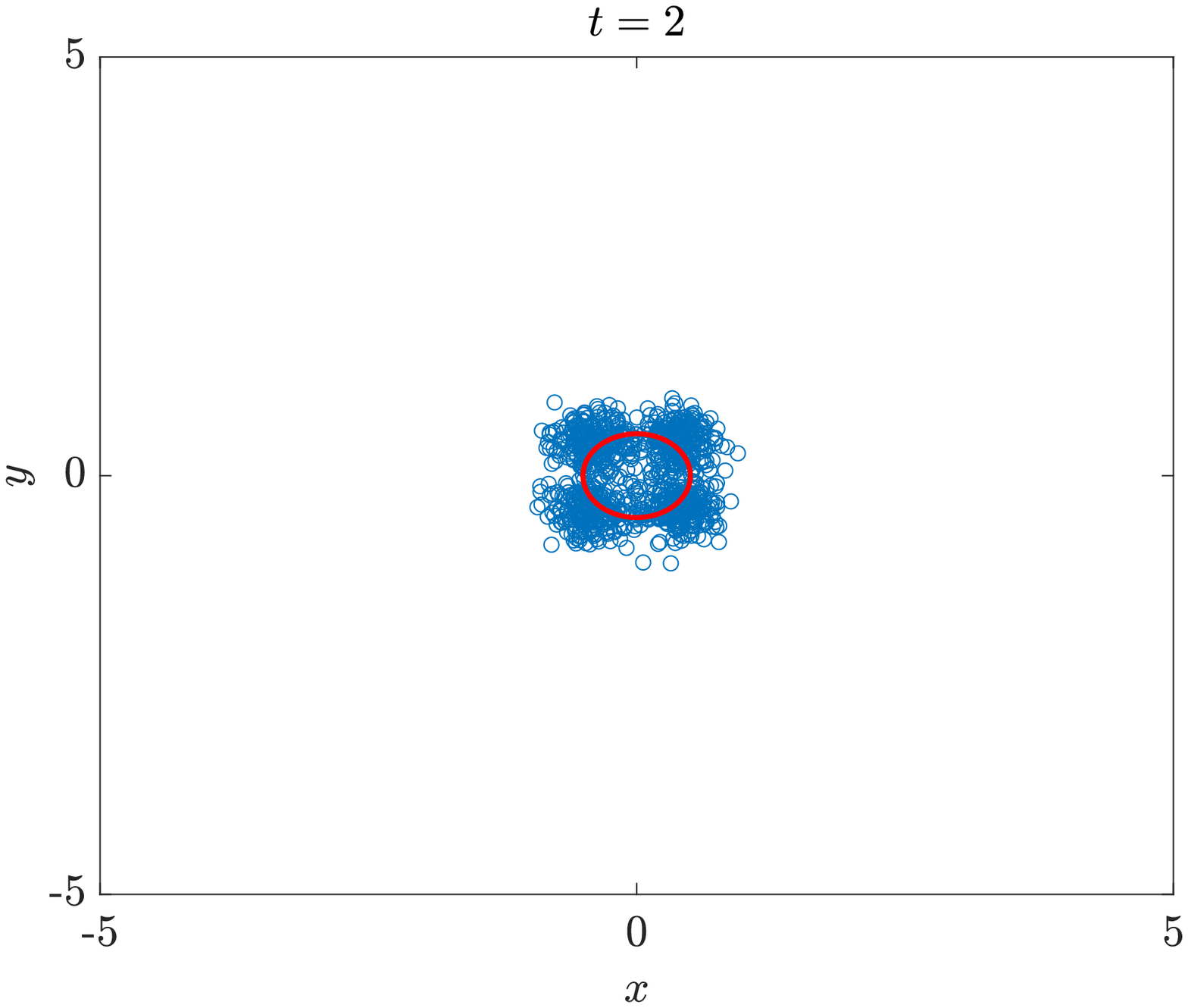}
\includegraphics[scale = 0.36]{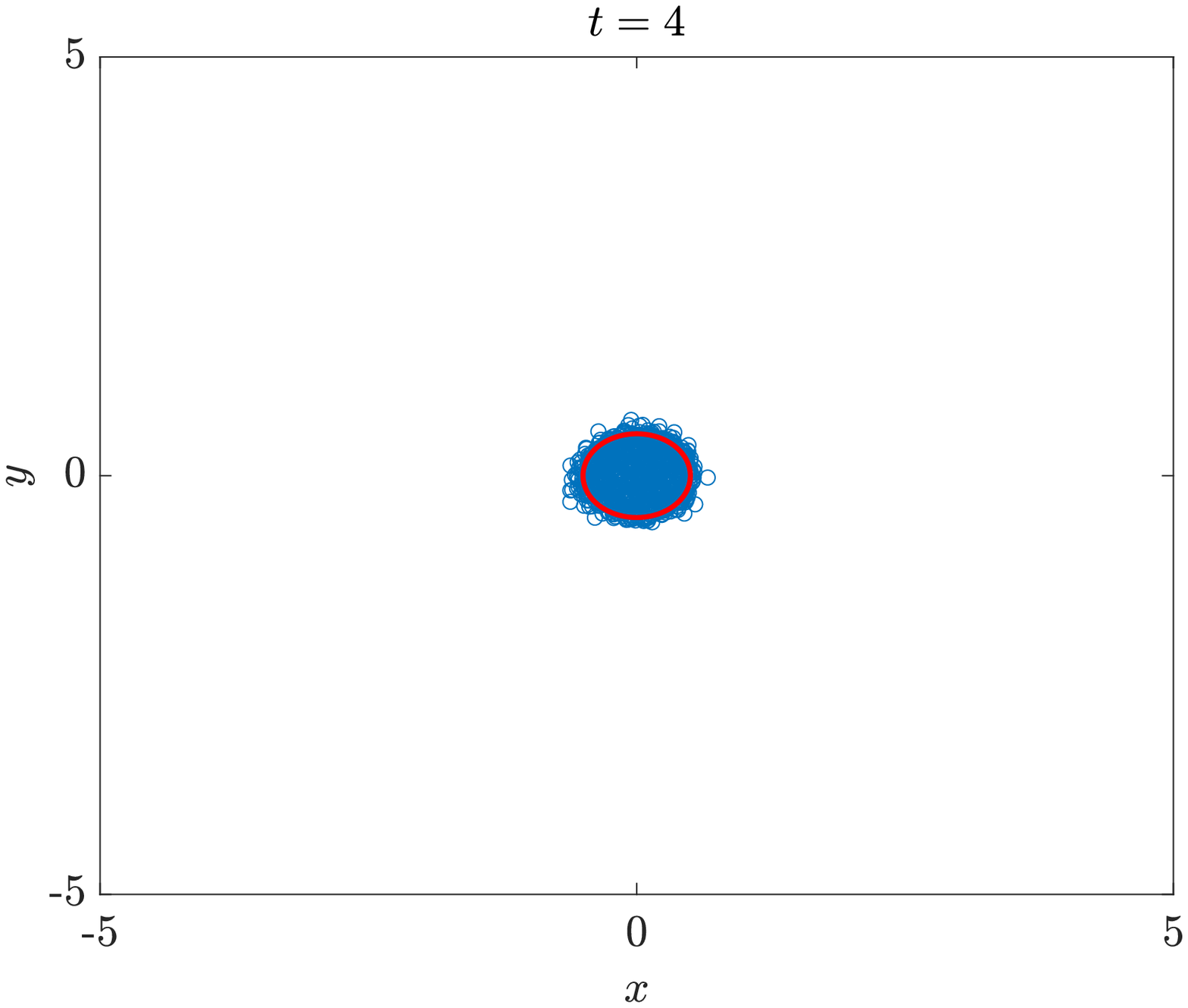}
\caption{\rev{Evolution of the particles' system \eqref{eq:SDE} in the 2D case and for $N = 10^3$ particles in the case $m_2 = 0.8$. We highlight in red the domain $D = \{\x \in \mathbb R^2:|\x| \le \frac{1}{2}\}$. The initial positions are sampled from \eqref{eq:f0_2D} .}}
\label{fig:part_2D}
\end{figure}
}

\section*{Conclusions}
In the present work, we introduced and studied  Fokker-Planck-type models suitable to describe the action of a large swarm on a fixed domain $D \subseteq \mathbb R^d$, with $d \ge 1$, being the main task of the swarm  to spread uniformly over its surface. This goal has be  obtained by  resorting  to a  description in terms of Fokker--Planck type equations with a  drift derived from a strongly convex quadratic potential, and a diffusion term with a variable coefficient of diffusion. The idea is that the drift  drives the particles towards  the center of the target domain $D$ and, as soon as it enters in $D$, each particle of the swarm starts a random exploration of the domain, with a speed that depends on its distance from the boundary. Recent mathematical results concerned with this type of Fokker--Planck equations allows to rigorously prove that the explicit asymptotic distribution profile,  a weighted combination of a uniform distribution inside $D$ and a Gaussian distribution in $D\setminus\mathbb R^d$, is reached in time at a polynomial rate. It is interesting to remark that the classical description in terms of the Fokker--Planck equation \fer{eq:model} with constant coefficient of diffusion and a drift which is derived by a potential that is not strongly convex, while leading to the same asymptotic profile, does not allow to use known mathematical results which ensure convergence of the solution towards equilibrium. A numerical approximation of the underlying Fokker-Planck equation in one and two dimensions allows to verify that convergence to the stationary solution holds at a certain polynomial rate, thus confirming the theoretical analysis. \rev{Several extensions of the present approach, which include weaker drift functions and dynamics on manifolds, are currently under study and will be presented in future works.  }

\section*{Acknowledgements}
This work has been written within the activities of the GNFM group of INdAM (National Institute of High Mathematics). M.Z. acknowledges partial support of MUR-PRIN2020 Project No. 2020JLWP23. The research of M.Z. was partially supported by MIUR, Dipartimenti di Eccellenza Program (2018–2022), and Department of Mathematics “F. Casorati”, University of Pavia.

\end{document}